%% file: main.tex
\documentclass[11pt]{article}
\usepackage[letterpaper, left=1in, right=1in, top=1in,bottom=1in]{geometry}
\usepackage{microtype}
\usepackage{parskip}

\usepackage[px]{dylan}
\usepackage{dylan}
\usepackage{color}              % Need the color package
\usepackage[suppress]{color-edits}
\addauthor{et}{magenta}
\addauthor{df}{magenta}
\addauthor{tl}{magenta}
\addauthor{ks}{magenta}
\addauthor{tlc}{cyan}
\addauthor{etc}{blue}
\addauthor{dfc}{red}
\addauthor{ksc}{green}

\usepackage[utf8]{inputenc} % allow utf-8 input
\usepackage[T1]{fontenc}    % use 8-bit T1 fonts
\usepackage{hyperref}       % hyperlinks
\usepackage{url}            % simple URL typesetting
\usepackage{booktabs}       % professional-quality tables
\usepackage{amsfonts}       % blackboard math symbols
\usepackage{nicefrac}       % compact symbols for 1/2, etc.
\begin{document}
\clearpage
% Title portion
%\title{Fast Convergence of Common Learning Algorithms in Games}
\title{Learning in Games: Robustness of Fast Convergence}

\date{}

\author{%
	Dylan J. Foster \thanks{Cornell University, \texttt{djfoster@cs.cornell.edu}. Work supported under NSF grant CDS\&E-MSS 1521544 and an NDSEG fellowship.}
	\and
    Zhiyuan Li\thanks{Tsinghua University, \texttt{lizhiyuan13@mails.tsinghua.edu.cn}. Research performed while author was visiting Cornell University.}
	\and 
	Thodoris Lykouris \thanks{Cornell University, \texttt{teddlyk@cs.cornell.edu}. Work supported under ONR grant N00014-08-1-0031, and a Google faculty research award.}
	\and
	Karthik Sridharan
	\thanks{Cornell University, \texttt{sridharan@cs.cornell.edu}. Work supported in part by NSF grant CDS\&E-MSS 1521544.}
	\and
	\'{E}va Tardos\thanks{Cornell University, \texttt{eva@cs.cornell.edu}. Work supported in part by NSF grant CCF-1563714, ONR grant N00014-08-1-0031, and a Google faculty research award.}
}
\maketitle
\thispagestyle{empty}

\input{abstract}
\newpage
\setcounter{page}{1}
\section{Introduction}
\input{intro}

%If players use $\eps$-l.a.r. algorithms, then something good happens.

\vspace{-0.1in} \section{Repeated Games and Learning Dynamics}\label{sec:model} \vspace{-0.1in}
%We will follow the notation of \cite{SyrgkanisALS15}. 
\input{model}

%\subsection{Learning Dynamics.} %\label{sec:low-approx-regret}
\input{learning_dynamics}

%\subsection{Approximate Efficiency of Outcomes of Low Approximate Regret Dynamics}
%\subsection{Smooth Games}
\input{approx_efficiency}

\section{Learning in Games with Full Information Feedback}\label{sec:static}

\input{full_info}

\section{Bandit 
Feedback}\label{sec:bandit}

\input{bandit}

\section{Dynamic Population Games}\label{sec:dynamic}

\input{dynamic}

%\section{Discussion}
%\begin{itemize}
%    \item Discuss possibility of bandit algorithms for the dynamic setting.
%    \item Optimal rates for playing with realizations?
%\end{itemize}

% \subsection*{Achknowledgements}
%\section{Discussion}\label{sec:discussion}
%\input{discussion}

%\subsubsection*{Acknowledgements}
%\tlccomment{Got a line here, maybe we want to add the grant support here due to space.}

\paragraph{Acknowledgements} We thank Vasilis Syrgkanis for sharing  simulation software and the NIPS reviewers for pointing out that the GREEN algorithm \cite{Allenberg2006} satisfies the Low Approximate Regret property.

\newpage

%% bibliography
%\setlength{\bibsep}{0.5pt}
{\small
\bibliographystyle{alpha}
\bibliography{bib1}
}
\newpage
\appendix
\section*{APPENDIX}
\setcounter{section}{0}
% \section{Proofs from Section \ref{sec:model}}\label{app:efficiency}
% \input{app_efficiency}

\section{Supplementary material for Section \ref{sec:static}}
\subsection{Proof of Proposition \ref{prop:poa_full_info-hp}}
\label{app:whp}
\input{app_whp}

\subsection{Low Approximate Property for Specific Algorithms}
\label{app:lar_full}
\input{app_approx_regret_algs}

\subsection{Proof of Theorem \ref{thm:static}}
\label{app_proof_main_theorem}
\input{app_proof_main_theorem}

\section{Supplementary material for Section \ref{sec:bandit}}
\label{app:bandit}
\subsection{Proof of Lemma
\ref{lem:bandit_static}}
\label{app:lar_bandit}
\input{app_proof_bandit_static}

\section{Supplementary Material for Section \ref{sec:dynamic}}\label{app:dynamic}

\subsection{Discussion of Results for Dynamic Population Games}
\label{turnover_results_dynamic_games}
\input{turnover_results_dynamic_games}

\subsection{Proof of Proposition \ref{prop:noisy_hedge}}
\label{app:lar_dynamic}
\input{app_dynamic}

\subsection{Proof of Proposition \ref{thm:poa_dynamic}}
\label{app_proof_dynamic_theorem}
\input{app_proof_dynamic_theorem}

\section{Utility Maximization Games and Mechanisms}
\label{app:utility}
\input{app-utility}

\end{document}

%% file: abstract.tex
% !TEX root = main.tex

\begin{abstract}
We show that learning algorithms satisfying a \emph{low approximate regret} property experience fast convergence to approximate optimality in a large class of repeated games. Our property, which simply requires that each learner has small regret compared to a $(1+\epsilon)$-multiplicative approximation to the best action in hindsight, is ubiquitous among learning algorithms; it is satisfied even by the vanilla Hedge forecaster. Our results improve upon recent work of Syrgkanis et al. \cite{SyrgkanisALS15} in a number of ways.
We require only that players observe payoffs under other players' realized actions, as opposed to expected payoffs. We further show that convergence occurs with high probability, and 
show convergence under bandit feedback. 
Finally, we improve upon the speed of convergence by a factor of $n$, the number of players. Both the scope of settings and the class of algorithms for which our analysis provides fast convergence are considerably broader than in previous work.

Our framework 
applies to dynamic population games via a low approximate regret property for shifting experts. Here we strengthen
the results of Lykouris et al. \cite{LykourisST16} in two ways: We allow players to select learning algorithms from a larger class, which includes a minor variant of the basic Hedge algorithm, and we increase the maximum churn in players for which approximate optimality is achieved.

In the bandit setting we present a new algorithm which provides a ``small loss''-type bound with improved dependence on the number of actions in utility settings, and is both simple and efficient. This result may be of independent interest. 
\end{abstract}

%% file: intro.tex
% !TEX root = main.tex

Consider players repeatedly playing a game, all acting independently to minimize their cost or maximize their utility. It is natural in this setting for each player to use a learning algorithm that guarantees small regret to decide on their strategy, as the environment is constantly changing due to each player's choice of strategy.
It is well known that such \emph{decentralized no-regret dynamics} are guaranteed to converge to a form of equilibrium for the game. 
Furthermore, in a large class of games known as \emph{smooth games} \cite{Roughgarden15}
they converge to outcomes with approximately optimal social welfare matching the worst-case efficiency loss of Nash equilibria (the \emph{price of anarchy}). In smooth cost minimization games the overall cost is $\lambda/(1-\mu)$ times the minimum cost, while in smooth mechanisms \cite{SyrgkanisT_STOC13} such as auctions it is $\lambda/\max(1,\mu)$ times the maximum total utility (where  $\lambda$ and $\mu$ are parameters of the smoothness condition). Examples of smooth games and mechanisms include routing games and many forms of auction games (see e.g. \cite{Roughgarden15,SyrgkanisT_STOC13,AuctionSurvey2016}).

The speed at which the game outcome converges to this approximately optimal welfare is governed by individual players' regret bounds. There are a large number of simple regret minimization algorithms (Hedge/Multiplicative Weights, Mirror Decent, Follow the Regularized Leader; see e.g. \cite{Hazan_book}) that guarantee that the average regret goes down as $O(1/\sqrt{T})$ with time $T$, which is tight in adversarial settings.

Taking advantage of the fact that playing a game against opponents who themselves are also using regret minimization is not a truly adversarial setting, a sequence of papers \cite{Daskalakis2015Near, Rakhlin2013Optimization, SyrgkanisALS15} showed that by using specific learning algorithms, the dependence on $T$ of the convergence rate can be improved to $O(1/T)$ (``fast convergence''). Concretely,  Syrgkanis et al. \cite{SyrgkanisALS15} show that all algorithms satisfying the so-called RVU property (Regret by Variation in Utilities), which include Optimistic Mirror Descent \cite{Rakhlin2013Optimization}, converge at a $O(1/T)$ rate 
with a fixed number of players. 

One issue with the works of \cite{Daskalakis2015Near, Rakhlin2013Optimization, SyrgkanisALS15} 
is that they use expected cost as their feedback model for the players. In each round every player receives the expected cost for each of their available actions, in expectation over the current action distributions of \emph{ all other players}.  This clearly represents more information than is realistically available to players in games --- at most each player sees the cost of each of their actions given the actions taken by the other players (\emph{realized feedback}). In fact, even if each player had access to the action distributions of the other players, simply computing this expectation is generally intractable when $n$, the number of players, is large. 

We improve the result of \cite{SyrgkanisALS15} on the convergence to approximate optimality
in smooth games in a number of different aspects. To achieve this, we relax the quality of approximation from the bound guaranteed by smoothness. Typical smoothness bounds on the price of anarchy in auctions are small constants, such a factor of 1.58 or 2 in item auctions. 
Increasing the approximation factor by an arbitrarily small constant $\eps > 0$ enables the following results: \vspace{-0.07in}
\begin{itemize}
\setlength{\itemsep}{0pt}\setlength{\parsep}{0pt}\setlength{\parskip}{0pt}
\item We show that learning algorithms obtaining
fast convergence
are ubiquitous.
\item We improve the speed of convergence by a factor of $n$, the number of players. 
\item For all
our results, players only need feedback based 
on realized --- not expected --- outcomes.
\item We show that convergence occurs with high probability in most settings.
\item We extend the results to show that it is enough for the players to observe realized \emph{bandit} feedback, only seeing the outcome of the action they play.
\item Our results apply to settings where the set of players in the game changes over time \cite{LykourisST16}. We strengthen previous results by showing that a broader class of algorithms achieve approximate efficiency under significant churn. \vspace{-0.07in}
\end{itemize}

We achieve these results 
using a property we term 
Low Approximate Regret, which simply states 
that an online learning algorithm achieves good regret against a multiplicative approximation of the best action in hindsight. This property is satisfied by many known algorithms including even the vanilla Hedge algorithm, as well as Optimistic Hedge \cite{RakhlinS13predictablesequences,SyrgkanisALS15} (via a new analysis). The crux of our analysis technique is the simple observation that for many types of data-dependent regret bounds we can fold part of the regret bound into the comparator term, allowing us to explore the trade-off between additive and multiplicative approximation. 

In Section \ref{sec:static}, we show that Low Approximate Regret implies fast convergence to the social welfare guaranteed by the price of anarchy via the smoothness property. This convergence only requires feedback from the realized
actions played by other players, not their action distribution or the expectation over their actions. We further show that this convergence occurs with high probability in most settings. For games with a large number of players we also improve the speed of convergence. \cite{SyrgkanisALS15} shows that players using Optimistic Hedge in a repeated game with $n$ players converge to the approximately optimal outcome guaranteed by smoothness at a rate of $O(n^2/T)$. They also offer an analysis guaranteeing $O(n/T)$ speed of convergence, at the expense of a constant factor decrease in the quality of approximation (e.g., a factor of 4 in atomic congestion games with affine congestion). We achieve the convergence bound of $O(n/T)$ with only an arbitrarily small loss in the approximation.

Algorithms that satisfy the Low Approximate Regret property are ubiquitous and include simple, efficient algorithms such as Hedge and variants.
The observation that this broad class of algorithms enjoys fast convergence in realistic settings suggests that fast convergence occurs in practice.

Comparing our work to \cite{SyrgkanisALS15} with regard to feedback,
%in more detail, 
%our Low Approximate Regret property applies more broadly than their RVU property. 
%Further, 
Low Approximate Regret algorithms require only realized feedback, while the 
analysis of the RVU property in  \cite{SyrgkanisALS15} requires expected feedback.
To see the contrast, consider the load balancing game introduced in \cite{Koutsoupias2009Worst} with two players and two bins, where each player selects a bin and observes cost given by the number of players in that bin. Initialized at the uniform distribution, any learning algorithm with expectation feedback (e.g. those in \cite{SyrgkanisALS15}) will stay at the uniform distribution forever, because the expected cost vector distributes cost equally across the two bins. This gives low regret under expected costs, but suppose we were interested in realized costs: The only ``black box'' way to lift \cite{SyrgkanisALS15} to this case would be to simply evaluate the regret bound above under realized costs, but here players will experience $\Theta(1/\sqrt{T})$ variation because they select bins uniformly at random, ruining the fast convergence. Our analysis sidesteps this issue because players achieve Low Approximate Regret with high probability.

In Section \ref{sec:bandit} we consider games where players can only observe the cost of the action they played  given the actions taken by the other players, and receive no feedback for actions not played (\emph{bandit feedback}). \cite{Rakhlin2013Optimization} analyzed zero-sum games with bandit feedback, but assumed that players receive expected cost over the strategies of all other players. In contrast, the Low Approximate Regret property can be satisfied by just observing realizations, even with bandit feedback. We propose a new bandit algorithm based on log-barrier regularization with importance sampling that guarantees fast convergence of $O(d\log T/\eps)$  where $d$ is the number of actions.  Known techniques would either result in a convergence rate of $O(d^3\log T)$ (e.g. adaptations of SCRiBLe \cite{RakhlinS13predictablesequences}) or would not extend to utility maximization settings (e.g. GREEN \cite{Allenberg2006}). Our technique is of independent interest since it improves the dependence of approximate regret bounds on the number of experts while applying to both cost minimization and utility maximization settings.

Finally, in Section \ref{sec:dynamic}, we consider the \emph{dynamic population game} setting of \cite{LykourisST16}, where players enter and leave the game over time. 
\cite{LykourisST16} showed that regret bounds for shifting experts directly influence the rate at which players can turn over and still guarantee close to optimal solutions on average. We show that a number of learning algorithms have the Low Approximate Regret property in the shifting experts setting, allowing us to extend the fast convergence result to dynamic games. Such learning algorithms include a noisy version of Hedge as well as AdaNormalHedge \cite{Luo2015}, which was previously studied in the dynamic setting in \cite{LykourisST16}.  
Low Approximate Regret allows us to increase the turnover rate from the
one in \cite{LykourisST16}, while also widening and simplifying the class of learning algorithms that players can use to guarantee the close to optimal average welfare. 

%% file: model.tex
% !TEX root = main.tex

We consider a 
game $G$ among a set of $n$ players. Each player $i$ has an action space $S_i$ and a cost function $\cost_i: S_1\times\dots\times S_n\rightarrow [0,1]$ that maps an action profile $s=(s_1,\dots,s_n)$ to the cost $\cost_i(s)$ that player experiences\footnote{See Appendix \ref{app:utility} for analogous definitions for utility maximization games.}.
We assume that the action space of each player has cardinality $d$, i.e. $|S_i|=d$. We let $w=(w_1,\dots,w_n)$ denote a list of probability distributions over all players' actions, where $w_i\in \Delta(S_i)$ and $w_{i,x}$ is the probability of action $x\in S_i$.

The game is repeated for $T$ rounds. At each round $t$ each player $i$ picks a probability distribution $w_i^t\in \Delta(S_i)$ over actions and draws their action $s_i^t$ from this distribution. Depending on the game playing environment under consideration, players will receive different types of feedback after each round.
In Sections \ref{sec:static} and \ref{sec:dynamic} we consider
feedback where at the end of the round each player $i$ observes the utility they would have received had they played any possible action $x\in S_i$ given the actions taken by the other players. More formally let $c_{i,x}^t= \cost_i(x,s_{-i}^t)$, where $s_{-i}^t$ is the set of strategies of all but the $i^{th}$ player at round $t$, and let $c_i^t=(c_{i,x}^t)_{x\in S_i}$.  Note that the expected cost of  player $i$ at round $t$ (conditioned on the other players' actions) is simply the inner product $\tri{w_i^t,c_i^t}$.

We refer to this form of feedback as \textit{realized feedback} since it only depends on the realized actions $s_{-i}^{t}$ sampled by the opponents; it does not directly depend on their distributions $w_{-i}^{t}$. This should be contrasted with the \textit{expectation feedback} used by  \cite{SyrgkanisALS15, Daskalakis2015Near, Rakhlin2013Optimization}, where player $i$ observes $\En_{s_{-i}^{t}\sim{}w_{-i}^{t}}\brk{\cost_i(x, s_{-i}^t)}$ for each $x$. 

Sections \ref{sec:bandit} and \ref{sec:dynamic} consider extensions of our repeated game model. In Section \ref{sec:bandit} we examine partial information (``bandit'') feedback, where players observe only the cost of their own realized actions. In Section \ref{sec:dynamic} we consider a setting where the player set is evolving over time. Here we use the dynamic population model of \cite{LykourisST16}, where at each round $t$ each player $i$ is replaced (``turns over'') with some probability $p$. The new player has cost function $\cost_i^t(\cdot)$ and action space $S_i^t$ which may change arbitrarily subject to certain constraints. We will formalize this notion later on.

%% file: learning_dynamics.tex
% !TEX root = main.tex
%\subsection
\paragraph{Learning Dynamics} \label{sec:low-approx-regret}
We assume that players select their actions using learning algorithms satisfying a property we call \emph{Low Approximate Regret}, which simply requires that the cumulative cost of the learner multiplicatively approximates the cost of the best action they could have chosen in hindsight. We will see in subsequent sections that this property is ubiquitous and leads to fast convergence in a robust range of settings.

\begin{definition}{(Low Approximate Regret)}
\label{def:lar}
A learning algorithm for player $i$ satisfies the \emph{Low Approximate Regret} property for parameter $\epsilon>0$ and function $A(d,T)$ if for all action distributions $f\in\Delta(S_i)$,
{\small
\begin{equation}
\label{eq:lar}
(1-\epsilon)\sum_{t=1}^T \tri{w_i^t,c_i^t}\leq \sum_{t=1}^T \tri{f,c_i^t}+\frac{A(d,T)}{\epsilon}.
\end{equation}}A learning algorithm satisfies Low Approximate Regret against shifting experts if for all sequences  $f^1,\dots, f^T\in\Delta(S_i)$, letting $K=|\{i>2: f^{t-1}\neq f^t \}|$ be the number of shifts,
{\small
\begin{equation}
\label{eq:lar_shifting}
(1-\epsilon)\sum_{t=1}^T \tri{w_i^t,c_i^t}\leq \sum_{t=1}^T\tri{f^t,c_i^t}+(1+K)\frac{A(d,T)}{\epsilon}.
\end{equation}
}In the bandit feedback setting, we require \eqref{eq:lar} to hold in expectation over the realized strategies of player $i$ for any $f\in\Delta(S_i)$ fixed before the game begins.
\end{definition}

We use the version of the Low Approximate Regret property with shifting experts when considering players in dynamic population games in Section \ref{sec:dynamic}. In this case, the game environment is constantly changing due to churn in the population, and we need the players to have low approximate regret with shifting experts to guarantee high social welfare despite the churn.

We emphasize that all algorithms we are aware of that satisfy Low Approximate Regret can be made to do so for any fixed choice of the approximation factor $\eps$ via an appropriate selection of parameters. Many algorithms have an even stronger property: They satisfy \eqref{eq:lar} or \eqref{eq:lar_shifting} \emph{for all $\eps>0$ simultaneously}. We say that such algorithms satisfy the \emph{Strong Low Approximate Regret} property. This property has favorable consequences in the context of repeated games.

The Low Approximate Regret property differs from previous properties such as RVU in that it only requires that the learner's cost be close to a multiplicative approximation to the cost of the best action in hindsight. Subsequently, it is always smaller than the regret. For instance, if we consider only uniform (i.e. not data-dependent) regret bounds the Hedge algorithm can only achieve $O\prn{\sqrt{T\log{}d}}$ exact regret, but can achieve Low Approximate Regret with parameters $\eps$ and $A(d,T)=O(\log{}d)$ for any $\eps>0$. Low Approximate Regret is analogous to the notion of $\alpha$-regret from \cite{KakadeSTOC07}, with $\alpha=(1+\epsilon)$.

In Appendix \ref{app:utility} we show that the Low Approximate Regret property and our subsequent results naturally extend to utility maximization games.

%% file: approx_efficiency.tex
% !TEX root = main.tex

\paragraph{Smooth Games} 
\label{sec:approx-efficiency}
It is well-known that in a large class of games, termed \emph{smooth games} by Roughgarden \cite{Roughgarden15}, traditional learning dynamics 
converge to approximately optimal social welfare. In subsequent sections we analyze the convergence of Low Approximate Regret learning dynamics in such smooth games. 
We will see that Low Approximate Regret (for sufficiently small $A(d,T)$) coupled with smoothness of the game implies fast convergence of  learning dynamics to desirable social welfare under a variety of conditions. Before proving this result we review social welfare and smooth games.

For a given action profile $s$, the social cost is $C(s)=\sum_{i=1}^n \cost_i(s)$. To bound the efficiency loss due to the selfish behavior of the players we define
{\small
\[
\opt=\min_{s^o}\sum_{i=1}^n \cost_i(s^o).
\]
}
\begin{definition}{(Smooth game \cite{Roughgarden15})}
\label{def:smoothness}
A cost minimization game is called $(\lambda,\mu)$-smooth if for all strategy profiles $s$ and $s^*$: $\sum_i \cost_i(s_i^*,s_{-i})\leq \lambda\cdot\cost_i(s^*)+\mu\cdot\cost_i(s)$.
\end{definition}
This property is typically applied using a (close to) optimal action profile $s^*=s^o$. For this case the property implies that if $s$ is an action profile with very high cost, then some player deviating to her share of the optimal profile $s_i^*$ will improve her cost. 

For smooth games, the price of anarchy is at most $\lambda/(1-\mu)$, meaning that Nash equilibria of the game, as well as no-regret learning outcomes in the limit, have social cost at most a factor of $\lambda/(1-\mu)$ above the optimum. Smooth cost minimization games include congestion games such as routing or load balancing. For example, atomic congestion games with affine cost functions are $(\frac{5}{3},\frac{1}{3})$-smooth \cite{Christodulou}, non-atomic games are $(1,0.25)$ smooth \cite{RoughgardenT2002}, implying a price of anarchy of 2.5 and 1.33 respectively.  While we focus on cost-minimization games for simplicity of exposition, an analogous definition also applies for utility maximization, including smooth mechanisms \cite{SyrgkanisT_STOC13}, which we elaborate on in Appendix \ref{app:utility}. Smooth mechanisms include most simple auctions. For example, the first price item auction is $(1-1/e,1)$-smooth and all-pay actions are $(1/2,1)$-smooth, implying a price of anarchy of 1.58 and 2 respectively. All of our results extend to such mechanisms.

%% file: full_info.tex
% !TEX root = main.tex

We now analyze the efficiency of algorithms with the Low Approximate Regret property in the full information setting. Our first proposition shows that, for smooth games with full information feedback, learners with the Low Approximate Regret property converge to efficient outcomes.

\begin{proposition}
\label{prop:poa_full_info}
In any $(\lambda,\mu)$-smooth game, if all players use Low Approximate Regret algorithms satisfying Eq. (\ref{eq:lar}) with parameters $\epsilon$ and $A(d,T)$, then for the action profiles $s^t$  drawn on round $t$ from the corresponding mixed actions of the players,
{\small
\[
\frac{1}{T}\sum_t\En\brk*{C(s^t)}\leq \frac{\lambda}{1-\mu-\epsilon}\opt+\frac{n}{T}\cdot\frac{1}{1-\mu-\epsilon}\cdot \frac{A(d,T)}{\epsilon}.
\]
}
\end{proposition}
\begin{proof}
This proof is a straightforward modification of the usual price of anarchy proof for smooth games. We obtain the claimed bound by writing
$\sum_t\En\brk*{C(s^t)}=\sum_i \sum_t\En\brk*{\cost_i(s^t)}$, using the Low Approximate Regret property with $f=s^*_i$ for each player $i$ for the optimal solution $s^*$, then using the smoothness property for each time $t$ to bound  $\sum_i \cost_i(s_i^*,s_{-i}^t)$, and finally rearranging terms.
\end{proof}

For $\epsilon<< (1-\mu)$ the approximation factor of $\lambda/(1-\mu-\epsilon)$ is very close to the price of anarchy $\lambda/(1-\mu)$. This shows that Low Approximate Regret learning dynamics quickly
converge to 
outcomes with 
social welfare 
arbitrarily close to 
the welfare guaranteed for exact Nash equilibria by the price of anarchy. A simple corollary of this proposition is that, when players use learning algorithms that satisfy the Strong Low Approximate Regret property, the bound above can be taken to depend on  $\opt$ even though this value is unknown to the players.

Whenever the Low Approximate Regret property is satisfied, a high probability version of the property with similar dependence on $\eps$ and $A(d,T)$ is also satisfied. This implies that in addition to quickly converging to efficient outcomes in expectation, Low Approximate Regret learners experience fast convergence with high probability.

\begin{proposition}
\label{prop:poa_full_info-hp}
In any $(\lambda,\mu)$-smooth game, if all players use Low Approximate Regret algorithms satisfying Eq. (\ref{eq:lar}) for parameters $\epsilon$ and $A(d,T)$, then for the action profile $s^t$  drawn on round $t$ from the players' mixed actions and $\gamma=2\eps/(1+\eps)$, we have that
$\forall \delta >0$, with probability at least $1 - \delta$,
{\small
$$
\frac{1}{T}\sum_t C(s^t)\leq \frac{\lambda}{1-\mu-\gamma}\opt+\frac{n}{T}\cdot\frac{1}{1-\mu-\gamma}\cdot \brk*{\frac{4A(d,T)}{\gamma}+\frac{12\log(n\log_2 (T)/\delta))}{\gamma}},
$$
}
\end{proposition}

\paragraph{Examples of Simple Low Approximate Regret Algorithms}
\label{sec:examples} 
Propositions \ref{prop:poa_full_info} and \ref{prop:poa_full_info-hp} are  informative when applied with algorithms for which $A(d,T)$ is sufficiently small. One would hope that such algorithms are relatively simple and easy to find. We show now that the well-known Hedge algorithm as well as basic variants such as Optimistic Hedge and Hedge with online learning rate tuning satisfy the property with $A(d,T)=O(\log{}d)$, which will lead to fast convergence both in terms of $n$ and $T$.
For these algorithms and indeed all that we consider in this paper, we can achieve the Low Approximate Regret property for any fixed $\eps>0$ via an appropriate parameter setting. In Appendix \ref{app:lar_full}, we provide full descriptions and proofs for these algorithms.

\begin{example}\label{thm:hedge_static}
Hedge satisfies the Low Approximate Regret property with $A(d,T) = \log(d)$. In particular one can achieve the property for any fixed $\eps>0$ by using $\eps$ as the learning rate.
\end{example}
\begin{example}\label{thm:hedge_tuned}
Hedge with online learning rate tuning
satisfies the Strong Low Approximate Regret property with $A(d,T) = O(\log{}d)$.
\end{example}
\begin{example}\label{thm:optimistic_hedge_static}
Optimistic Hedge %slightly differs from Hedge in the update rule, slightly giving more weight to the last observation. Let  $g^{t+1}\propto g^{t}e^{-\eta c^t}$ be the Hedge update rule. Then the Optimistic Hedge update rule is: $w^{t+1}\propto g^{t+1}e^{-\eta c^{t+1}}$
% \begin{itemize}
% \item Optimistic Hedge 
%with a constant learning rate %$\eta=\epsilon/8<1/4$
satisfies the Low Approximate Regret property with $ A(d,T) =  8 \log(d)$. As with vanilla Hedge, we can choose the learning rate to achieve the property with any $\eps$. 
\end{example}
\begin{example}
\label{ex:small_loss}
Any algorithm satisfying a ``small loss'' regret bound of the form $\sqrt{(\textrm{Learner's cost})\cdot{}A}$ or $\sqrt{(\textrm{Cost of best action})\cdot{}A}$ satisfies Strong Low Approximate Regret via the AM-GM inequality, i.e. $\sqrt{(\textrm{Learner's cost})\cdot{}A} \propto \inf_{\eps>0}\brk*{\eps\cdot{}(\textrm{Learner's cost}) + A/\eps}$.
In particular, this implies that the following algorithms have Strong Low Approximate Regret: Canonical small loss and self-confident algorithms, e.g. \cite{Freund97,Auer2002Adaptive,Yaroshinsky2004}, Algorithm of \cite{Cesa2007Improved}, Variation MW \cite{Hazan2010Extracting}, AEG-Path \cite{Steinhardt2014Adaptivity}, AdaNormalHedge \cite{Luo2015}, Squint \cite{Koolen2015Second}, and Optimistic PAC-Bayes \cite{Foster2015Adaptive}.
\end{example}

Example \ref{ex:small_loss} shows that the Strong Low Approximate Regret property in fact is ubiquitous, as it is satisfied by any algorithm that provides small loss regret bounds or one of many variants on this type of bound. 
Moreover, all algorithms that satisfy the Low Approximate Regret property for all fixed $\eps$ can be made to satisfy the strong property using the doubling trick.

\paragraph{Main Result for Full Information Games:}
\label{comparison_Vasilis} 

\begin{theorem}
\label{thm:static}
In any $(\lambda,\mu)$-smooth game, if all players use Low Approximate Regret algorithms satisfying (\ref{eq:lar}) for parameter $\epsilon$ \footnote{We can also show that the theorem  holds if players satisfy the property for different values of $\eps$, but with a dependence on the worst case value of $\eps$ across all players.} and $A(d,T)=O(\log d)$, then
{\small
$$
\frac{1}{T}\sum_t\En\brk*{C(s^t)}\leq \frac{\lambda}{1-\mu-\epsilon}\opt+\frac{n}{T}\cdot\frac{1}{1-\mu-\epsilon}\cdot \frac{O(\log d)}{\epsilon},
$$
}
and furthermore, $\forall \delta >0$, with probability at least $1 - \delta$,
{\small
$$
\frac{1}{T}\sum_t\En\brk*{C(s^t)}\leq \frac{\lambda}{1-\mu-\epsilon}\opt+\frac{n}{T}\cdot\frac{1}{1-\mu-\epsilon}\cdot \brk*{\frac{O(\log d)}{\epsilon}+\frac{O(\log(n\log_2 (T)/\delta))}{\epsilon}}.
$$
}
\end{theorem}
\begin{corollary}
\label{corr:strong}
If all players use Strong Low Approximate Regret algorithms then:
1. The above results hold for all $\epsilon>0$ simultaneously.
2. Individual players have regret bounded by $O\prn{T^{-1/2}}$, even in adversarial settings.
3. The players approach a coarse correlated equilibrium asymptotically.
\end{corollary}

\paragraph{Comparison with Syrgkanis et al. \cite{SyrgkanisALS15}.}
By relaxing the standard $\lambda/(1-\mu)$ price of anarchy bound, Theorem \ref{thm:static} substantially broadens the class of algorithms that experience fast convergence to include even the common Hedge algorithm. The main result of \cite{SyrgkanisALS15} shows that learning algorithms that satisfy their RVU property  converge to the price of anarchy bound $\lambda/(1-\mu)$ at the rate of  $n^2\log{}d/T$. They further show how to achieve a worse approximation of $\lambda(1+\mu)/(\mu(1-\mu))$ at the improved (in terms of $n$) rate of  $n \log{}d/T$. 
We converge to an approximation arbitrarily close to $\lambda/(1+\mu)$ at a rate of $n \log{}d/T$. Note that in atomic congestion games with affine congestion function $\mu=1/3$, so their bound of $\lambda(1+\mu)/\mu(1-\mu)$ loses a factor of 4 compared to the price of anarchy.

Strong Low Approximate Regret algorithms such as Hedge with online learning rate tuning simultaneously experience both fast $O(n/T)$ convergence in games and an $O(1/\sqrt{T})$
bound on individual regret in adversarial settings. In contrast, \cite{SyrgkanisALS15} only shows $O(n/\sqrt{T})$ individual regret and $O(n^{3}/T)$ convergence to price of anarchy simultaneously.

Low Approximate Regret algorithms only need realized feedback, whereas \cite{SyrgkanisALS15} require expectation feedback. Having players receive expectation feedback is  unrealistic in terms of both information and computation. Indeed, even if the necessary information was available, computing expectations over discrete probability distributions is not tractable in the general case unless $n$ is taken to be constant.

Our results imply that Optimistic Hedge enjoys the best of two worlds: It enjoys fast convergence to the exact $\lambda/(1-\mu)$ price of anarchy using expectation feedback as well as fast convergence to the $\eps$-approximate price of anarchy using realized feedback. Our new analysis of Optimistic Hedge (Appendix \ref{app:proof_optimistic_hedge_static}) sheds light on another desirable property of this algorithm: Its regret is bounded in terms of the net cost incurred by Hedge. Figure \ref{fig:comparison} summarizes the differences between our results.

\begin{figure}[h!]
\centering
{\small
\begin{tabular}{|l|l|l|l|l|}
\hline
 & Feedback & POA & Rate & Time comp.  \\
 \hline
     \begin{tabular}{l}RVU property~{\small\cite{SyrgkanisALS15}}\end{tabular}& \begin{tabular}{l}Expected costs\end{tabular} & exact & $O(n^2 \log d/T)$ & $d^{O(n)}$ per round\\
     \hline
     \begin{tabular}{l}LAR property ~ {\small(section \ref{sec:model})}\end{tabular} & \begin{tabular}{l}Realized costs\end{tabular} & $\epsilon$-approx & $O(n \log d/(\epsilon T))$ & $O(d)$ per round\\
     \hline
\end{tabular}}
\caption{Comparison of Low Approximate Regret and RVU properties.\vspace{-0.1in}}
\label{fig:comparison}
\end{figure}

%% file: bandit.tex
% !TEX root = main.tex

In many realistic scenarios, the players of a game might not even know what they would have lost or gained if they had deviated from the action they played.  
We model this lack of information with \emph{bandit feedback}, in which each player observes a single scalar, $\mathbf{cost}_i(s^t) = \tri{s^t_i, c_i^t}$, per round.\footnote{With slight abuse of notation, $s^t_i$ denotes the identity vector associated to the strategy player $i$ used at time $t$.}
When the game considered is smooth, one can use the Low Approximate Regret property as in the full information setting to show that players quickly converge to efficient outcomes. Our results here hold with the same generality as in the full information setting: As long as learners satisfy the Low Approximate Regret property \eqref{eq:lar}, an efficiency result analogous to Proposition \ref{prop:poa_full_info} holds.

\begin{proposition}
\label{thm:bandit}
Consider a $(\lambda,\mu)$-smooth game.  If all players use bandit learning algorithms with Low Approximate Regret $A(d,T)$ then
{\small
\[
\frac{1}{T}\En\brk*{\sum_tC(s^t)}\leq \frac{\lambda}{1-\mu-\epsilon}\opt+\frac{n}{T}\cdot\frac{1}{1-\mu-\epsilon}\cdot \frac{A(d,T)}{\epsilon}.
\]
}
\end{proposition}

\paragraph{Bandit Algorithms with Low Approximate Regret}
The bandit Low Approximate Regret property requires that \eqref{eq:lar} holds in expectation against any sequence of adaptive and potentially adversarially chosen costs, but only for an obliviously chosen comparator $f$.\footnote{This is because we only need to evaluate \eqref{eq:lar} with the game's optimal solution $s^{\star}$ to prove efficiency results.} This is weaker than requiring that an algorithm achieve a true expected regret bound; it is closer to pseudo-regret.

The Exp3Light algorithm \cite{Stoltz05} satisfies Low Approximate Regret with $A(d,T) = d^2\log{}T$. The \scrible{} algorithm introduced in \cite{Abernethy08competingin} (via the analysis in \cite{RakhlinS13predictablesequences}) enjoys
the Low Approximate Regret property with $A(d,T) = d^3 \log(dT)$. The GREEN algorithm introduced in \cite{Allenberg2006} achieves the Low Approximate Regret property with $A(d,T)=d\log(T)$, but only works with costs and not gains. This prevents it from being used in utility settings such as auctions, as in Appendix \ref{app:utility}.

We present a new bandit algorithm (Algorithm \ref{alg:bandit_alg}) that achieves Low Approximate Regret with $A(d,T) = d \log(T/d)$ and thus matches the performance of GREEN, but works in both cost minimization and utility maximization settings. This method is based on Online Mirror Descent with a logarithmic barrier for the positive orthant, but differs from earlier algorithms based on the logarithmic barrier (e.g. \cite{RakhlinS13predictablesequences}) in that it uses the classical importance-weighted estimator for costs instead of sampling based on the Dikin elipsoid. It can be implemented in $\tilde{O}(d)$ time per round, using line search to find $\gamma$. We provide the proof of this lemma and further discussion of Algorithm \ref{alg:bandit_alg} in Appendix \ref{app:bandit}.

\textbf{Algorithm 3: }{\bf Initialize $w^1$  to the uniform distribution. On each round $t$, perform update: }
{\small
\begin{align}\label{alg:bandit_alg}
\textbf{Algorithm  \ref{alg:bandit_alg} update: }~~~~~{\bf w^t_{s^{t-1}} = \frac{w^{t-1}_{s^{t-1}}}{1 + \eta c^t_{s^{t-1}} + \gamma w^{t-1}_{s^{t-1}}} ~~~ \textrm{and}~~~\forall j \ne s^{t-1}~~  w^t_{j} = \frac{w^{t-1}_{j}}{1  + \gamma w^{t-1}_{j}}},
\end{align}}{\bf where $\gamma \le0$ is chosen so that $w^t$ is a valid probability distribution.
}
\begin{lemma}\label{lem:bandit_static}
Algorithm \ref{alg:bandit_alg} with $\eta =\epsilon/(1+\epsilon)$ has Low Approximate Regret with $A(d,T)=O(d\log T)$.
\end{lemma}
\paragraph{Comparison to Other Algorithms}
In contrast to the full information setting where the most common algorithm, Hedge, achieves Low Approximate Regret with competitive parameters, the most common 
adversarial bandit algorithm \expthree{} does not seem to satisfy Low Approximate Regret.
\cite{Alletal06} provide a small loss bound for bandits which would be sufficient for Low Approximate Regret, but their algorithm requires prior knowledge on the loss of the best action (or a bound on it), which is not appropriate in our game setting. Similarly, the small loss bound in \cite{Neu15} is not applicable in our setting as the work assumes an oblivious adversary and so does not apply to the games we consider.

%% file: dynamic.tex
% !TEX root = main.tex

In this section we consider the dynamic population repeated game setting introduced in \cite{LykourisST16}. Detailed discussion and all proofs are deferred to Appendix \ref{app:dynamic}.
Given a game $G$ as described in Section \ref{sec:model}, a \emph{dynamic population game} with \emph{stage game} $G$ is a repeated game where at each round $t$ game $G$ is played and every player $i$ is replaced by a new player with a \emph{turnover probability} $p$. Concretely, when a player turns over, their strategy set and cost function are changed arbitrarily subject to the rules of the game. This models a repeated game setting where players have to adapt to an adversarially changing environment. We denote the cost function of player $i$ at round $t$ as $\cost_i^t(\cdot)$. As in Section \ref{sec:static}, we assume that the players receive full information feedback. At the end of each round they observe the entire cost vector $c_i^t=\cost_i^t(\cdot,s_{-i}^t)$, but are not aware of the costs of the other players in the game. 

\paragraph{Learning in Dynamic Population Games and the Price of Anarchy} To guarantee small overall cost using the smoothness analysis from Section \ref{sec:approx-efficiency}, players need to exhibit low regret against a shifting benchmark $s_i^{*t}$ of socially optimal strategies achieving $\opt^t=\min_{s^{*t}}\sum_i \cost_i^t(s^{*t})$. 
Even with a small probability $p$ of change, the sequence of optimal solutions can have too many changes to be able to achieve low regret. In spite of this apparent difficulty, \cite{LykourisST16} prove that at least a $\rho\lambda/(1-\mu-\epsilon)$ fraction of the optimal welfare is guaranteed if 
1. players are using low \textit{adaptive regret} 
algorithms (see \cite{Hazan2009, Luo2015}) and 2. for the underlying optimization problem there exists a relatively \emph{stable} sequence of solutions which at each step approximate the optimal solution by a factor of $\rho$. This holds as long as the turnover probability $p$ is upper bounded by a function of $\epsilon$ (and of certain other properties of the game, such as the stability of the close to optimal solution).

We consider dynamic population games where each player uses a learning algorithm satisfying Low Approximate Regret for shifting experts \eqref{eq:lar_shifting}. This shifting version of Low Approximate Regret implies a dynamic game analog of our main efficiency result, Proposition \ref{prop:poa_full_info}.
 
\paragraph{Algorithms with Low Approximate Regret for Shifting Experts} A simple variant of Hedge we term \emph{Noisy Hedge}, which mixes the Hedge update at each round with a small amount of uniform noise, satisfies the Low Approximate Regret property for shifting experts with $A(d,T)=O(\log (dT))$. Moreover, algorithms that satisfy a small loss version of the \emph{adaptive regret} property \cite{Hazan2009} used in \cite{LykourisST16} satisfy the Strong Low Approximate Regret property. 

\begin{proposition}\label{prop:noisy_hedge}
Noisy Hedge with learning rate $\eta=\epsilon$ satisfies the Low Approximate Regret property for shifting experts with $A(d,T)=2\log(dT)$.
\end{proposition}

\paragraph{Extending Proposition \ref{prop:poa_full_info} to the Dynamic Population Game Setting}
Let $s^{*1:T}$ denote a stable sequence of near-optimal solutions $s^{*t}$ with $\sum_i \cost_i^t(s^{*t}) \le \rho\cdot \opt^t$ for all rounds $t$. As discussed in \cite{LykourisST16}, such stable sequences can come from simple greedy algorithms (where each change in the input of one player affects the output of few other players) or via differentially private algorithms (where each change in the input of one player affects the output of all other players with small probability); in the latter case the sequence is randomized.
For a deterministic sequence  $s_i^{*1:T}$ of player $i$'s actions, we let the random variable $K_i$ denote the number of changes in the sequence. For a randomized sequence $s_i^{*1:T}$, we let $K_i$ be the sum of total variation distances between subsequent pairs $s_i^{*t-1}$ and $s_i^{*t}$. The stability of a sequence of solutions is determined by $\mathbb{E}[\sum_i K_i]$. 
\begin{proposition}{(PoA with Dynamic Population)}
\label{thm:poa_dynamic}
If all players use Low Approximate Regret algorithms satisfying (\ref{eq:lar_shifting}) in a dynamic population game, where the stage game is $(\lambda,\mu)$-smooth, and $K_i$
as defined above then 
{\small
\begin{equation}
\label{eq:dynamic_efficiency}
\frac{1}{T}\sum_t \En\brk*{C(s^t)}\leq \frac{1}{T} \frac{\lambda\cdot \rho}{1-\mu-\epsilon}\sum_t\En\brk*{\opt^t}+\frac{n+\En\brk*{\sum_i K_i}}{T}\cdot\frac{1}{1-\mu-\epsilon}\cdot\frac{A(d,T)}{\epsilon}.
\end{equation}}Here the expectation is taken over the random turnover in the population playing the game, as well as the random choices of the players on the left hand side. 
\end{proposition}

To claim a price of anarchy bound, we need to ensure that the additive term in \eqref{eq:dynamic_efficiency} is a small fraction of the optimal cost. The challenge is that high turnover probability %$p$
reduces stability, increasing $\En\brk*{\sum_i K_i}$. By using algorithms with smaller $A(d,T)$, we can allow for higher $\En\brk*{\sum_i K_i}$ and hence higher turnover probability. Combining Noisy Hedge with Proposition \ref{thm:poa_dynamic} 
strengthens the results in \cite{LykourisST16} by both weakening the behavioral assumption on the players, allowing them to use simpler learning algorithms, and allowing a higher turnover probability.

\paragraph{Comparison to Previous Results}
\cite{LykourisST16} use the more complex  AdaNormalHedge algorithm of \cite{Luo2015} which satisfies the adaptive regret property of \cite{Hazan2009}, but has $O(dT)$ space complexity. In contrast, Noisy Hedge only requires space complexity of just $O(d)$. Moreover, a broader class of algorithms satisfy the Low Approximate Regret property which makes the efficiency guarantees more prescriptive since this property serves as a behavioral assumption. Finally, the guarantees we provide improve on the turnover probability that can be accommodated. We provide further discussion in Appendix 
\ref{turnover_results_dynamic_games}.

%% file: app_whp.tex
% !TEX root = main.tex
\paragraph{Overview}
The core of Proposition \ref{prop:poa_full_info-hp} is Lemma \ref{theorem:full_info_high_prob}, which shows that the regret of an individual player concentrates around its expectation. To prove the high probability efficiency result of Proposition \ref{prop:poa_full_info-hp} we simply apply this lemma to the individual players, apply the union bound to get a regret statement that holds for all players simultaneously, then finally apply the same smoothness argument used in the expectation case.

\begin{lemma}[High-probability regret bound]
\label{theorem:full_info_high_prob}
Let $w^t\in\Delta(d)$ be selected by an algorithm satisfying the Low Approximate Regret property \eqref{eq:lar} for $\eps>0$ given costs $c^t$ selected by an adaptive adversary, and let $s^t\sim{}w^t$ be the algorithm's realized action. 
Then for all $\delta\in(0,\min\crl*{1,n\log{}_2{}T/e})$ and $T\geq{}4$, with probability at least $1-\delta$,
\begin{equation}
\label{eq:tdh_high_prob}
(1-\gamma)\sum_{t=1}^{T}\tri{s^t,c^t}\leq{} \sum_{t=1}^{T}\tri*{s^{\star},c^t} + \frac{4A(d,T)}{\gamma} + \frac{12\log(\log_2(T)/\delta)}{\gamma},
\end{equation}
where $\eps=\gamma/(2-\gamma)$.
\end{lemma}
Before proving Lemma \ref{theorem:full_info_high_prob} we restate a refinement of Freedman's martingale Bernstein inequality due to \cite{Bartlett2008High} which is a standard tool for proving high-probability versions of data-dependent regret bounds.
\begin{lemma}[\cite{Bartlett2008High}]
\label{theorem:freedman}
Let $X_{1},\ldots,X_{T}$ be a martingale difference sequence with $\abs{X_{t}}\leq{}b$. Let $\bar{\sigma}^{2}=\sum_{t=1}^{T}\textrm{Var}(X_{t}\mid{}X_1,\ldots,X_{t-1})$ be the sum of conditional variances for a particular outcome $X_1,\ldots,X_T$. For all $\delta\in(0,1/e)$, $T\geq{}4$ we have{\small
\begin{equation}
\label{eq:refined_freedman}
\Pr\prn*{\sum_{t=1}^{T}X_{t} > 4\sqrt{\bar{\sigma}^{2}\log(1/\delta)} + 2b\log(1/\delta)} \leq{} \log_{2}(T)\delta.
\end{equation}}
\end{lemma}

\begin{proof}[Proof of Lemma \ref{theorem:full_info_high_prob}]
Let $Z^{t}(s^1,\ldots,s^{t}) = (1-\epsilon)\tri{s^t, c^t}$
be a random process indexed by $t\in\brk{T}$.
We leave the dependence of $c^{t}$ on $s^1,\ldots,s^{t-1}$ implicit. 
Let $X^t(s^1,\ldots,s^t)=Z^t(s^1,\ldots,s^t) - \En\brk*{Z^t\mid{}s^1,\ldots,s^{t-1}}$ be the associated martingale difference sequence. 
Note that $\abs{X^t}\leq{}(1-\epsilon)\leq{}1$ and that $\En\brk{s^t\mid{}s^1,\ldots,s^{t-1}} = w^t$.

Lemma \ref{theorem:freedman} applied to $\sum_{t=1}^{T}X^t(s^1,\ldots,s^t)$ and the Low Approximate Regret property \eqref{eq:lar} now imply that for a given draw of $s^1,\ldots,s^{T}$, with probability at least $1-\delta$,
\begin{align*}
(1-\epsilon)\sum_{t=1}^{T}\tri{s^t,c^t} =\sum_{t=1}^{T}Z_t &\leq{} \sum_{t=1}^{T}\En\brk*{Z^t\mid{}s^1,\ldots,s^{t-1}}+ 4\sqrt{\bar{\sigma}^{2}\log(\log_2(T)/\delta)} + 2\log(\log_2(T)/\delta)\\
&= (1-\eps)\sum_{t=1}^{T}\tri*{w^{t},c^t}+ 4\sqrt{\bar{\sigma}^{2}\log(\log_2(T)/\delta)} + 2\log(\log_2(T)/\delta)\\
&\leq \sum_{t=1}^{T}\tri*{s^{\star},c^t}+\frac{A(d,T)}{\eps}+ 4\sqrt{\bar{\sigma}^{2}\log(\log_2(T)/\delta)} + 2\log(\log_2(T)/\delta),
\end{align*}
where $s^{\star}=\argmin_{i\in\brk{d}}\sum_{t=1}^{T}\tri*{e_i,c^t}$.
To complete this bound we must provide a bound on the conditional variance $\bar{\sigma}^{2}$. To this end note that
\begin{align*}
\bar{\sigma}^{2} &= \sum_{t=1}^{T}\En\brk*{\prn*{X_s^t}^2\mid{}s^1,\ldots,s^{t-1}}\\ &= \prn*{1-\epsilon}^{2}\sum_{t=1}^{T}\En\brk*{\prn*{\tri{s^t,c^t} - \tri{w^t, c^t}}^2\mid{}s^1,\ldots,s^{t-1}}.
\intertext{Now, since the mean minimizes the squared error:}
&\leq \prn*{1-\epsilon}^{2}\sum_{t=1}^{T}\En\brk*{\prn*{\tri{s^t,c^t}}^2\mid{}s^1,\ldots,s^{t-1}}.
\end{align*}
Since $c^t\in\brk{0,1}^{d}$ we have
\begin{align*}
\bar{\sigma}^2&\leq \prn*{1-\epsilon}^{2}\sum_{t=1}^{T}\En\brk*{\tri{s^t,c^t}\mid{}s^1,\ldots,s^{t-1}}\\
&= \prn*{1-\epsilon}^{2}\sum_{t=1}^{T}\tri{w^t, c^t}.
\end{align*}
% Finally, by union bound over $s$ we have that with probability at least $1-\delta$,
Hence, with probability at least $1-\delta$,
\begin{align*}
&\prn*{1-\epsilon}\sum_{t=1}^{T}\tri{s^t,c^t} - \sum_{t=1}^{T}\tri*{s^{\star},c^t}\\ &\leq{} \frac{A(d,T)}{\epsilon} + 4\sqrt{\prn*{1-\epsilon}^{2}\prn*{\sum_{t=1}^{T}\tri{w^t, c^t}}\log\prn*{\log_2(T)/\delta}} + 2\log\prn*{\log_2(T)/\delta}.\\
\intertext{Now for all $\eps'>0$ by the AM-GM inequality we have:}
&\leq{} \frac{A(d,T)}{\epsilon} + \eps'\prn*{1-\epsilon}^{2}\prn*{\sum_{t=1}^{T}\tri{w^t, c^t}} + 4\log\prn*{\log_2(T)/\delta}/\eps' + 2\log\prn*{\log_2(T)/\delta},
\intertext{and so the Low Approximate Regret property \eqref{eq:lar} implies}
&\leq{} \frac{A(d,T)}{\epsilon} + \eps'(1-\epsilon)\brk*{\sum_{t=1}^{T}\tri*{s^{\star},c^t} + \frac{A(d,T)}{\eps}
}+ 4\log\prn*{\log_2(T)/\delta}/\eps' + 2\log\prn*{\log_2(T)/\delta}.\\
\end{align*}
Rearranging,
\begin{align*}
&\frac{(1-\epsilon)}{(1+\eps'(1-\eps))}\sum_{t=1}^{T}\tri{s^t,c^t}\\
&\leq{}\sum_{t=1}^{T}\tri*{s^{\star},c^t}+ \frac{1}{1+\eps'(1-\eps)}\brk*{\frac{A(d,T)}{\epsilon} +
4\log\prn*{\log_2(T)/\delta}/\eps' + 2\log\prn*{\log_2(T)/\delta}} + \frac{A(d,T)}{\eps}.
\end{align*}
Taking $\eps'=\eps$  we have
\[
\frac{(1-\eps)}{1+\eps-\eps^{2}}\sum_{t=1}^{T}\tri{s^t,c^t}\leq{} \sum_{t=1}^{T}\tri*{s^{\star},c^t} + 2\frac{A(d,T)}{\eps} + \frac{6\log(\log_2(T)/\delta)}{\eps}.
\]
We simplify to a slightly weaker bound,
\[
\frac{(1-\eps)}{(1+\eps)}\sum_{t=1}^{T}\tri{s^t,c^t}\leq{} \sum_{t=1}^{T}\tri*{s^{\star},c^t} + 2\frac{A(d,T)}{\eps} + \frac{6\log(\log_2(T)/\delta)}{\eps}.
\]
Now setting $\eps = \gamma/(1-\gamma)$ we arrive at
\[
(1-2\gamma)\sum_{t=1}^{T}\tri{s^t,c^t}\leq{} \sum_{t=1}^{T}\tri*{s^{\star},c^t} + 2\frac{A(d,T)}{\gamma} + \frac{6\log(\log_2(T)/\delta)}{\gamma}.
\]

Finally, reparameterizing with $\gamma'=2\gamma$ we have

\[
(1-\gamma')\sum_{t=1}^{T}\tri{s^t,c^t}\leq{} \sum_{t=1}^{T}\tri*{s^{\star},c^t} + \frac{4A(d,T)}{\gamma'} + \frac{12\log(\log_2(T)/\delta)}{\gamma'}.
\]
\end{proof}

%% file: app_approx_regret_algs.tex
% !TEX root = main.tex

In this section, we present the proofs of the Low Approximate Regret property for Hedge (Example \ref{thm:hedge_static}) and Optimistic Hedge (Example \ref{thm:optimistic_hedge_static}). The first proof is only included for completeness but may be helpful as subsequent proofs follow the same framework. Our proof for Optimistic Hedge includes a new analysis that relates the performance of Optimistic Hedge on a given cost sequence to the performance of Hedge on the same sequence. The analysis shows that Optimistic Hedge will experience low regret whenever Hedge has low cost, which in particular implies that it satisfies the Low Approximate Regret property.
We omit the proof for Example \ref{thm:hedge_tuned} and instead refer the reader to Corollary 2.4 of \cite{prediction_book}, which derives the result using the doubling trick.

\subsubsection{Hedge (Example
\ref{thm:hedge_static})}\label{app:proof_hedge_static}
Hedge is an algorithm for online linear optimization over the simplex $\Delta(d)$. It has update rule
\[
w^{t+1}_{i} \propto w^{t}_{i}e^{-\eta c^t_{i}}\quad\forall{}i\in\brk{d},
\] where $\eta>0$ is the \emph{learning rate}.

%Hedge can be understood 
We derive Hedge as an instance of Online Mirror Descent (see e.g. \cite{Hazan_book}) with the \emph{negative entropy regularizer} $R(w)=\sum_{i=1}^d w_i \log(w_i)$.
To run Online Mirror Descent one picks a learning rate $\eta>0$ and initial weights (also known as a \emph{prior}) $w^1$, then performs the following update step at each time $t\in\brk{T}$:
\begin{enumerate}
\item Let $\widetilde{w}^t$ satisfy $\nabla R(\widetilde{w}^{t+1})=\nabla R(w^t)-\eta c^t$.
\item $w^{t+1}=\argmin_{f\in\Delta(d)} D_R(f|\widetilde{w}^{t+1})$.
\end{enumerate}
Here $D_R(f|g)\defeq{}R(f)-R(g)-\tri*{\nabla R(g),f-g}$ is the \emph{Bregman divergence} for the regularizer $R$. We briefly restate some useful properties of the Mirror Descent update.
\begin{lemma}[Properties of Mirror Descent (e.g. \cite{Hazan_book})]\label{lem:property_of_projection}
For any convex regularizer $R$ we have
\begin{itemize}
\item $\breg(f\mid{}g)\geq{}0$.
\item For any $a,b,c\in\Delta(d)$,
\[
\tri*{b-a,\nabla{}R(c)-\nabla{}R(b)}=D_{R}(a\mid{}c) - D_{R}(a\mid{}b)-D_{R}(b\mid{}c).
\]
\item The Mirror Descent update can alternatively be expressed as
\[
w^{t+1} = \argmin_{f\in\Delta(d)}\eta{}\tri{f,c^t}+\breg\prn*{f\mid{}w^t}.
\]
\item Any update of the form $f^* = \argmin_{f\in\Delta(d)}\tri*{f,c}+\breg\prn*{f\mid{}w}$ satisfies
\begin{equation*}
\tri*{f^*-g,c}\leq \breg\prn*{g\mid{}w}-\breg\prn*{g\mid{}f^*}-\breg\prn*{f^*\mid{}w}\quad\forall{}g\in\Delta(d).
\end{equation*}
\end{itemize}
\end{lemma}
%}

\begin{proposition}\label{prop:hedge}
\textit{Hedge, when run with constant learning rate and uniform prior $w^{1}_{i}=1/d$, satisfies the Low Approximate Regret property with $A(d,T)=\log(d)$.}
\end{proposition}

\begin{proof}[Proof of Proposition \ref{prop:hedge}]
Using the standard Online Mirror Descent analysis we have that at every step $t$, for any $f\in \Delta(d)$:
\begin{align} 
    \tri{w^t-f,c^t}&\le \tri{w^t-\widetilde{w}^{t+1},c^t}+ \frac{1}{\eta} \prn*{D_R(f|w^t)-D_R(f|w^{t+1}) - D_R(\wt^{t+1}|w^t)} \notag \\
    & \le  \tri{w^t-\widetilde{w}^{t+1},c^t}+ \frac{1}{\eta} \prn*{D_R(f|w^t)-D_R(f|w^{t+1}) }. \label{eq:hedge_per_step}
\end{align}
For the first term in the sum above, we have:
\begin{equation}\label{eq:hedge_per_step_first_term}
\tri{w^t-\widetilde{w}^{t+1},c^t}\leq \eta \tri{w^t,c^t}
\end{equation}
To see this note that $\nabla R(w)=\log(w)+1$ (where $\log$ is applied element-wise) and hence $(\nabla R)^{-1}(f)=e^{f-1}$. This %\dfcdelete{implies that $(\nabla R)^{-1}(f)=e^{f-1}$ which} 
implies $\widetilde{w}_i^{t+1}=w_i^t e^{-\eta c_i^t}$, and so
\begin{equation}\label{eq:hedge_loss_place}\tri{w^t-\widetilde{w}^{t+1},c^t}=\sum_{j\in[d]} w_j^t c_j^t(1-e^{-\eta c_j^t})  \leq \eta\sum_{j\in[d]} w_j^t(c_j^t)^2\leq \eta \tri{w^t,c^t}.
\end{equation}
The first inequality in \eqref{eq:hedge_loss_place} uses  that $1-e^{-\eta x}\leq \eta x$ for $x>0$ and the second inequality uses that the losses lie in $[0,1]$.

Using relations \eqref{eq:hedge_per_step} and \eqref{eq:hedge_per_step_first_term}, and summing over $t$:
\begin{equation}
    \sum_t\tri{w^t-f,c^t}\leq  \eta\sum_t \tri{w^t,c^t}+\frac{1}{\eta}D_R(f|w^1).
\end{equation}
Since $w^1$ is the uniform distribution, $D_R(f|w^1)\leq \log(d)$. Rearranging yields the claimed result.
\end{proof}

\subsubsection{Optimistic Hedge
(Example \ref{thm:optimistic_hedge_static})}\label{app:proof_optimistic_hedge_static}

The Optimistic Hedge algorithm performs two separate weight updates at each timestep to produce its action distribution. The method first performs a Hedge update $g_i^{t+1}\propto g_i^{t}e^{-\eta c_i^t}$, then produces the prediction distribution: $w_i^{t+1}\propto g_i^{t+1}e^{-\eta c_i^{t}}$.

\begin{lemma}
\label{lem:optimistic_hedge}
\textit{Optimistic Hedge with a constant learning rate $\eta=\epsilon/8<1/4$ satisfies the Low Approximate Regret property with $A(d,T)=8\log(d)$.}
\end{lemma}

Let $R$ be the negative entropy regularizer as in the proof of Proposition \ref{prop:hedge}. Let $\nabla^2 R(w)$ denote the Hessian of the regularizer $R$. The \emph{local norm} with respect to $w$ is $\nrm{f}_{w}=\sqrt{f^T \nabla^2 R(w) f}$ and its dual norm is $\nrm{x}_{w}^{\star}=\sqrt{x^T(\nabla^2 R(w))^{-1}x}$. 
For the negative entropy regularizer this definition yields $\nrm{f}_{w}^{2}=\sum_{i\in\brk{d}}\frac{\prn*{f_i}^{2}}{w_i}$ and $(\nrm{x}_{w}^{\star})^2=\sum_{i\in\brk{d}}w_i(x_i)^{2}$.
%We will now introduce a useful lemma due to
We begin by restating an intermediate Lemma from \cite{RakhlinS13predictablesequences} that bounds the regret of Optimistic Hedge in terms of the local norm.
\begin{lemma}{(Lemma 3 in \cite{RakhlinS13predictablesequences})}\label{lem:predictable}
Optimistic Hedge enjoys for any $f\in \Delta(S)$
\begin{equation}
\label{eq:rs13}
\sum_{t=1}^T\tri{w^t-f,c^t}\leq 2\eta \sum_{t=1}^T (\nrm{c^t-c^{t-1}}_{w^t}^{\star})^2+\frac{\log(d)}{\eta}.
\end{equation}
as long as $\eta \nrm{c^t-c^{t-1}}_{\infty}\leq 1/4$ at every step.
\end{lemma}
\begin{proof}[Proof of Lemma \ref{lem:optimistic_hedge}]
We will focus on the first term in the right-hand side of \eqref{eq:rs13} and prove that for all $t$,
\begin{equation}\label{eq:optimistic_auxiliary}
(\nrm{c^t-c^{t-1}}_{w^t}^{\star})^2\leq 2 \tri{w^t,c^t}+2 \tri{g^{t-1},c^{t-1}}.
\end{equation}
This holds as
\begin{align}
(\nrm{c^t-c^{t-1}}_{w^t}^{\star})^2 &\leq 2\big((\nrm{c^t}_{w^t}^{\star})^2 + (\nrm{c^{t-1}}_{w^t}^{\star})^2\big)\notag\\
& = 2\big(\sum_{j=1}^d w_j^t (c_j^t)^2+ \sum_{j=1}^d w_j^t (c_j^{t-1})^2)\big)\label{eq:oh1}\\
& \leq 2\big(\tri{w^t,c^t}+\tri{w^t,c^{t-1}}\big)\label{eq:oh2}\\
&= 2\big(\tri{w^t,c^t}+\tri{g^{t-1},c^{t-1}}+\tri{w^{t}-g^{t},c^{t-1}}+\tri{g^{t}-g^{t-1},c^{t-1}}\big)\notag\\
& \leq 2\big(\tri{w^t,c^t}+\tri{g^{t-1},c^{t-1}}\big)\label{eq:oh3}.
\end{align}
Here \eqref{eq:oh1} holds by the definition of the local norm, \eqref{eq:oh2} holds as the costs are in $[0,1]$, and \eqref{eq:oh3} holds via two applications of Lemma \ref{lem:property_of_projection} for Bregman projections:
\[
\tri{w^{t}-g^{t},c^{t-1}}\leq{} D_{R}(g^t\mid{}g^t) - D_{R}(g^t\mid{}w^t) - D_{R}(w^t\mid{}g^t)\leq{}0.
\]
\[
\tri{g^{t}-g^{t-1},c^{t-1}}\leq{} D_R(g^{t-1}\mid{}g^{t-1}) - D_R(g^{t-1}\mid{}g^t) - D_R(g^t\mid{}g^{t-1})\leq{}0.
\]

Now, applying \eqref{eq:optimistic_auxiliary} to Lemma \ref{lem:predictable}, we have that for $\eta<1/4$,
\begin{equation}\label{eq:optimistic_auxiliary_2}
\sum_{t=1}^T\tri{w^t-f,c^t}\leq 4\eta \sum_{t=1}^T\tri{w^t,c^t}+4\eta \sum_{t=1}^T\tri{g^{t-1},c^{t-1}}+\frac{\log(d)}{\eta}.
\end{equation}
Observe now that $g^t$ are the weights selected by the basic Hedge algorithm on the sequence $\crl{c^t}$ (setting $c^0=0$ and $g^0$ uniform). Hence by the Low Approximate Regret property for Hedge (Example \ref{thm:hedge_static}) we have
$$
\sum_{t=1}^T\tri{w^t-f,c^t}\leq 4\eta \sum_{t=1}^T\tri{w^t,c^t}+\frac{4\eta}{1-\eta} \Big(\sum_{t=1}^{T}\tri{f,c^{t-1}}+\frac{\log(d)}{\eta }\Big) +\frac{\log(d)}{\eta}.
$$
Rearranging,
$$
(1-4\eta)\sum_{t=1}^T\tri{w^t,c^t}\leq \frac{1+3\eta}{1-\eta}\big(\sum_{t=1}^{T}\tri{f,c^t}+\frac{\log(d)}{\eta}\big).
$$
This gives the claimed bound as $1+3\eta\leq \frac{1}{1-3\eta}$ for $\eta\leq 1/3$ and $1-8\eta\leq (1-4\eta)(1-\eta)(1-3\eta)$.
\end{proof}
%\dfcedit{

%% file: app_proof_main_theorem.tex
% !TEX root = main.tex
Theorem \ref{thm:static} follows immediately from Propositions \ref{prop:poa_full_info} and \ref{prop:poa_full_info-hp}. 

Corrollary \ref{corr:strong} holds because the Strong Low Approximate Regret property states that \eqref{eq:lar} holds for all $\epsilon>0$, so in particular we can set $\epsilon=\sqrt{\frac{\log(d)}{T}}$ to arrive at the desired regret bound.

%% file: app_proof_bandit_static.tex
% !TEX root = main.tex

Algorithm \ref{alg:bandit_alg} follows a standard design scheme for bandit algorithms. First we develop an algorithm with a full information regret bound, then run this algorithm using an unbiased estimator for the cost.

Let $R(w) = \sum_{j\in\brk{d}}\log(1/w_j)$; we call this the log barrier regularizer because it is a logarithmic barrier for the positive orthant. Algorithm \ref{alg:bandit_alg} is equivalent to the following update step at each time $t$:
\begin{itemize}
\item Sample $s^{t}\sim{}w^t$.
\item Observe $c^t_{s^t}$ and build the importance-weighted estimator:
$$\hat{c}_j^t=   \begin{cases} c_j^t/w_j^t &\mbox{if } j=s^t \\ 
0 & \mbox{otherwise}\end{cases}.
$$
\item Update $w^{t+1}$ with a Mirror Descent step from $\hat{c}^t$:
\begin{enumerate}
\item Let $\widetilde{w}^t$ satisfy $\nabla R(\widetilde{w}^{t+1})=\nabla R(w^t)-\eta \hat{c}^t$.
\item $w^{t+1}=\argmin_{f\in\Delta(d)} D_R(f|\widetilde{w}^{t+1})$.
\end{enumerate}
\end{itemize}

Note that $\hat{c}^t$ is \emph{unbiased} in that it satisfies $\En_{s^t \sim w^t}\brk*{\hat{c}^t}=c^t$.

\paragraph{Overview}
In this section we state and prove Lemma \ref{lem:bandit_final}, which provides a regret bound for Algorithm \ref{alg:bandit_alg}. To prove Lemma \ref{lem:bandit_static}, we apply Lemma \ref{lem:bandit_final} with learning rate $\eta = \epsilon/(1+\epsilon)$ and observe that the Low Approximate Regret property is satisfied:
\begin{align*}
(1 - \epsilon)\En\brk*{\sum_t \tri{e_{s^t},c^t}} \le \sum_t \tri{f,c^t}+ \frac{d (1 + \epsilon) \log(T/d)}{\epsilon} + d.
\end{align*}
In Appendix \ref{app:utility} we sketch a proof of the regret bound for Algorithm \ref{alg:bandit_alg} in the case where utilities are used instead of costs.
\paragraph{Regret Bound for Full Information} 
\begin{proposition}[Properties of the log barrier regularizer]
\label{prop:log_barrier}
Recall that $R(w) = \sum_{i\in\brk{d}}\log(1/w_i)$.
\begin{itemize}
\item $\nabla R(w)=-1/w$, which implies that for all $i$: $\widetilde{w}_i^{t+1}=\frac{w_i^t}{1+\eta w_i^t \hat{c}_i^t}$.
\item $D_{R}(f\mid{}w) = \sum_{i\in\brk{d}}\brk*{\log\prn*{\frac{w_i}{f_i}} + \frac{f_i}{w_i}}-d$.
\end{itemize}
\end{proposition}
\begin{lemma}
\label{lem:bandit_full_info}
Online Mirror Descent (see section \ref{app:proof_hedge_static}) with the log barrier regularizer, for any sequence of costs $c^{1},\ldots,c^{T}$ in $\mathbb{R}^{d}$, produces weights $w^t$ that satisfy the following bound for any $f^{\star}\in\Delta(d)$:
\begin{equation}
\label{bandit_full_info}
\tri{w^t-f^{\star},c^t}  \leq  \eta\sum_{j\in\brk{d}} \frac{(w_j^t \cdot c_j^t)^2}{1+\eta w_j^t c_j^t}
    +\frac{1}{\eta}\prn*{D_R\prn*{f^{\star}\mid{}w^t}-D_R\prn*{f\mid{}w^{t+1}}}.
\end{equation}
In particular, it achieves the regret bound
\begin{equation}
\label{eq:bandit_regret_pointwise}
\sum_{t=1}^{T}\tri{w^t-f^{\star},c^t}  \leq  \eta  \sum_{t=1}^{T}\sum_{j\in\brk{d}} \frac{(w_j^t \cdot c_j^t)^2}{1+\eta w_j^t c_j^t}
    +\frac{1}{\eta}D_R\prn*{f^{\star}\mid{}w^1}.
\end{equation}
\end{lemma}
\begin{proof}[Proof of Lemma \ref{lem:bandit_full_info}]
Fix $f^{\star}\in\Delta(d)$. 
Starting from the standard Mirror Descent proof we have that for each $t$:
\begin{align*}
    \tri{w^t-f^{\star},\hat{c}^t} &\leq 
    \tri{w^t-\widetilde{w}^{t+1},\hat{c}^t}
    +\frac{1}{\eta}\prn*{D_R\prn*{f^{\star}\mid{}w^t}-D_R\prn*{f^{\star}\mid{}w^{t+1}}}.
\intertext{The result is obtained by plugging in the expression for $\wt^t$ from Proposition \ref{prop:log_barrier}:}
    &=\eta\sum_j \frac{ (w_j^t \hat{c}_j^t)^2}{1+\eta w_j^t \hat{c}_j^t}+\frac{1}{\eta}\prn*{D_R\prn*{f^{\star}\mid{}w^t}-D_R\prn*{f^{\star}\mid{}w^{t+1}}}.
\end{align*}
The regret bound is obtained by summing this inequality.
\end{proof}
 
\paragraph{From Full Information to Partial Information}

\begin{lemma}[Regret bound for Algorithm \ref{alg:bandit_alg}]
\label{lem:bandit_final}
For any $f^{\star}\in\Delta(d)$, and any sequence of costs $c^{1},\ldots,c^{T}\in\brk*{0,1}^{d}$, the weights generated by Algorithm \ref{alg:bandit_alg} with $\eta\in\prn*{0,1}$ satisfy
\
\begin{equation}\label{bandit_final}
\En\brk*{\sum_{t=1}^T \tri{w^t-f^{\star},c^t}} \leq \frac{\eta}{1 -\eta} \En\brk*{\sum_{t=1}^T \tri{w^t,c^t}}
    +\frac{1}{\eta}d\log(T/d) + d.
\end{equation}
\end{lemma}

\begin{proof}[Proof of Lemma \ref{lem:bandit_final}]
Observe that Algorithm \ref{alg:bandit_alg} is equivalent to running Online Mirror Descent with $R$, using the unbiased estimator $\hat{c}^t$ for costs, where we recall $\hat{c}^t_i=\ind\crl*{s^t=i}c^t_i/w^t_i$. Thus, Lemma \ref{lem:bandit_full_info} implies that at each time $t$,
\begin{align*}
\tri{w^t-f^{\star},\hat{c}^t} &\leq \eta\sum_{j\in\brk{d}} \frac{(w_j^t \cdot \hat{c}_j^t)^2}{1+\eta w_j^t \hat{c}_j^t}
    +\frac{1}{\eta}\prn*{D_R\prn*{f^{\star}\mid{}w^t}-D_R\prn*{f^{\star}\mid{}w^{t+1}}}.
    \intertext{Since $w^t_j\hat{c}^t_j=\ind\crl*{s^t=j}c^t_j$, we have:}
    &= \eta\sum_{j\in\brk{d}}\ind\crl{s^t=j}\frac{(c_j^t)^2}{1+\eta c_j^t}
    +\frac{1}{\eta}\prn*{D_R\prn*{f^{\star}\mid{}w^t}-D_R\prn*{f^{\star}\mid{}w^{t+1}}}
\end{align*}
Taking the conditional expectation of each side of this inequality we have
\begin{align*}
    \tri{w^t-f^{\star},c^t} &= \En\brk*{\tri{w^t-f^{\star},\hat{c}^t}\mid{}s^1,\ldots,s^{t-1}}\\
    &\leq{} \En\brk*{\eta\sum_{j\in\brk{d}}\ind\crl{s^t=j}\frac{(c_j^t)^2}{1+\eta c_j^t}\mid{}s^1,\ldots,s^{t-1}} 
    +\frac{1}{\eta}\prn*{D_R\prn*{f^{\star}\mid{}w^t}-D_R\prn*{f^{\star}\mid{}w^{t+1}}}\\
     &= \eta\sum_{j\in\brk{d}}\frac{w^t_j(c_j^t)^2}{1+\eta c_j^t}  +\frac{1}{\eta}\prn*{D_R\prn*{f^{\star}\mid{}w^t}-D_R\prn*{f^{\star}\mid{}w^{t+1}}}.\\
     \intertext{Now, since $c^t_j$ lie in the range $\brk{0,1}$ we have}
     &\leq{} \frac{\eta}{1-\eta}\tri*{w^t,c^t}  +\frac{1}{\eta}\prn*{D_R\prn*{f^{\star}\mid{}w^t}-D_R\prn*{f^{\star}\mid{}w^{t+1}}}.
\end{align*}

Summing over all $t$ and taking a final expectation yields the bound,
\begin{equation}
\label{eq:bandit_dr}
\En\brk*{\sum_{t=1}^T \tri{w^t-f^{\star},c^t}} \leq \frac{\eta}{1 -\eta} \En\brk*{\sum_{t=1}^T \tri{w^t,c^t}}
    +\frac{1}{\eta}D_R\prn*{f^{\star}\mid{}w^1}.
\end{equation}

It remains to bound the Bregman divergence term. A direct approach fails here because one can choose $f^{\star}$ to make $D_R(f^{\star}\mid{}w^1)$ arbitrarily large; this is in contrast with the case where $D_R$ is the KL divergence, where we have a $\log{}d$ bound as long as $w^1$ is uniform.
To sidestep this difficulty, given arbitrary $f^{\star}\in\Delta(d)$ we let $\bar{f}=(1-\theta)f^{\star}  + \theta\pi$, where $\theta\in\brk{0,1}$ and $\pi$ is the uniform distribution. By Proposition \ref{prop:log_barrier} we have
\begin{equation}\label{bandit_auxiliary}
D_R\prn*{\bar{f}\mid{}w^1}\leq d\log(1/\theta)
\end{equation}

Applying \eqref{eq:bandit_dr} with $\bar{f}$ as the comparator now implies
\[
\En\brk*{\sum_{t=1}^T \tri{w^t-\bar{f},c^t}} \leq \frac{\eta}{1 -\eta} \En\brk*{\sum_{t=1}^T \tri{w^t,c^t}}
    +\frac{1}{\eta}d\log(1/\theta).
\]

Rearranging, this implies
\begin{align*}
    \En\brk*{\sum_{t=1}^T \tri{w^t-f^{\star},c^t}} &\leq \frac{\eta}{1 -\eta} \En\brk*{\sum_{t=1}^T \tri{w^t,c^t}}
    + \theta\sum_{t=1}^{T}\tri*{\pi, c^t} +\frac{1}{\eta}d\log(1/\theta).
    \intertext{Since we have assumed $c^t\in\brk{0,1}^{d}$, this is bounded as}
    &\leq \frac{\eta}{1 -\eta} \En\brk*{\sum_{t=1}^T \tri{w^t,c^t}}
    + \theta{}T +\frac{1}{\eta}d\log(1/\theta).
    \intertext{Finally, setting $\theta=d/T$ yields the desired bound:}
    &\leq \frac{\eta}{1 -\eta} \En\brk*{\sum_{t=1}^T \tri{w^t,c^t}}
     +\frac{1}{\eta}d\log(T/d) + d.
\end{align*}
\end{proof}

%% file: turnover_results_dynamic_games.tex
% !TEX root = main.tex
We briefly show how Proposition \ref{thm:poa_dynamic} with players using Low Approximate Regret algorithms 
with $A(d,T)=O(\log(dT))$ improves the maximum turnover rate $p$ in the results of \cite{LykourisST16}. 

In Definition \ref{def:lar}, Low Approximate Regret for shifting experts \eqref{eq:lar_shifting} is defined in terms of the number of shifts $K=|\{i>2: f^{t-1}\neq f^t \}|$ in a sequence of comparators $f^1,\ldots, f^T$. To compare with \cite{LykourisST16} we need a slightly different notion of Low Approximate Regret based on the \emph{total variation distance} of the sequence $f^1,\ldots, f^T$. Letting $K=\sum_t\nrm*{f^t-f^{t-1}}_1$, we require
\begin{equation}
\label{eq:lar_tv}
(1-\epsilon)\sum_{t=1}^T \tri{w_i^t,c_i^t}\leq \sum_{t=1}^T\tri{f^t,c_i^t}+(1+K)\frac{A(d,T)}{\epsilon}.
\end{equation}
In fact, whenever Low Approximate Regret for shifting experts \eqref{eq:lar_shifting} holds, \eqref{eq:lar_tv} holds as well as explained in \cite{LykourisST16}. Thus, without loss of generality we take $K_{i}$ to be the total variation distance of the solution sequence $s_{i}^{*1:T}$ for the $i$th player going forward, since if player $i$ satisfies Low Approximate Regret for shifting experts \eqref{eq:lar_shifting} they also satisfy:
\begin{equation}
(1-\epsilon)\sum_{t=1}^T \tri{w_i^t,c_i^t}\leq \sum_{t=1}^T\tri{s^{*t}_i,c_i^t}+(1+\sum_{t=2}^{T}\nrm{s_{i}^{*t}-s_{i}^{*t-1}}_{1})\frac{A(d,T)}{\epsilon}.
\end{equation}

Let $\kappa$ denote the expected number of players whose strategy in $s^{*1:T}$ changes as one player turns over, so $\En\brk*{\sum_i K_i} =pnT \kappa$ (as in expectation $pn$ players turn over at each step). The parameter $\kappa$ as defined here depends on the concrete game; it is a parameter of a high stability approximate optimization method used as in \cite{LykourisST16}. Let $\gamma>0$ be a lower bound on the minimum cost of each player, at each time step, so that we have $\sum_t\En\brk*{\opt^t}\ge \gamma nT$. Using the two parameters $\kappa$ and $\gamma$, \cite{LykourisST16} show a price of anarchy bound of $\lambda \rho/(1-\mu-\epsilon)$, assuming the turnover probability $p$ satsifies $p \le \epsilon^2\gamma^2/(\kappa\log(dT))$. Using Proposition \ref{thm:poa_dynamic} with with $A(d,T)=O(\log(dT))$ (as in, for example, Noisy Hedge) we get the same price of anarchy bound, yet allow higher turnover probability by a factor of $1/\gamma$:  We tolerate $p \le \epsilon^2\gamma/(\kappa\log(dT))$.

To illustrate this improvement, consider matching markets. Suppose $n$ players are each bidding in a first price item auction for one of many items (i.e., the winner pays her own bid for each item). Further suppose $v_{ij}$, the player $i$'s value for item $j$, has $v_{ij}\in [\gamma,1]$, and that the players are unit-demand, each bidding with the goal of winning one high value item at a low price. In this mechanism, we will use  $SW(s)$ to denote the social welfare achieved by action profile $s$, the sum of player utilities plus the auctioneer's revenue, and use $\opt^t$ is the maximum social welfare possible with players in round $t$.

The first price item auction is a $(1-1/e,1)$ smooth mechanism and hence has a price of anarchy of $e/(e-1)\approx 1.58$. Lykouris et al. \cite{LykourisST16} prove that a price of anarchy of %$1.58(1+\eps)$ 
$3.16(1+\eps)$ is guaranteed if players use adaptive learning and the turnover probability is at most $p \le \eps^2\gamma^2/\prn*{\log (dT)\log(1/\gamma)}$, which corresponds to $\rho=2$ and $\kappa=\log(1/\gamma)$. Using the proof from \cite{LykourisST16} with the improved $A(d,T)$ term of Proposition \ref{prop:noisy_hedge}, we get an $\gamma^{-1}$ improvement in the probability term.

\begin{theorem}
If all players use Low Approximate Regret algorithms for shifting experts with parameters $\eta$ and $A(d,T)=\log(dT)$
in a dynamic population matching market with first price item auctions, then 
{\small
\begin{equation}
3.16(1+\eta)\sum_t \En\brk*{SW(s^t)}\geq  \sum_t\En\brk*{\opt^t},
\end{equation}
}
assuming the turnover probability $p$ has at most $p\le \eps^2\gamma/\prn*{\log(dT)\log (1/\gamma)}$. 
\end{theorem}
In other games and  mechanisms \cite{LykourisST16} including congestion games, bandwidth-sharing, and large markets, we achieve analogous improvements.

%% file: app_dynamic.tex
% !TEX root = main.tex
Noisy Hedge is a modification of Hedge that mixes the distribution returned by the exponential update with a small uniform noise at each step. Fix $\theta\in\brk{0,1}$, $\eta>0$, and let $\pi$ be the uniform distribution over $\brk{d}$. Let $w^1=\pi$. Then the Noisy Hedge update at time $t$ is given by:
\begin{enumerate}
\item $\wt^{t+1}_{i} = w^t_ie^{-\eta{}c^t_i}$.
\item $g^{t+1}_i = \wt^{t+1}_{i}/\sum_{j\in\brk{d}}\wt^{t+1}_{j}$.
\item $w^{t+1}= (1-\theta)g^{t+1} + \theta{}\pi$.
\end{enumerate}

\begin{lemma}
\label{lem:noisy_hedge}
Let $f^{1},\ldots,f^{T}\in\Delta(d)$ be any sequence of experts with $K$ changes. Then for any sequence of costs $c^1,\ldots,c^{T}\in\brk{0,1}^{d}$, Noisy Hedge with learning rate $\eta>0$ and $\theta=1/T$ enjoys the regret bound
\[
\sum_{t=1}^T \tri{w^t-f^t,c^t}\leq{}\eta\sum_{t=1}^{T}\tri{w^t,c^t} + \frac{1}{\eta}\prn*{2\log{}d + K\log(dT)}.
\]
\end{lemma}
\begin{proof}[Proof of Lemma \ref{lem:noisy_hedge}]
We follow a proof similar to that of Hedge (Proposition \ref{prop:hedge}). Note that we have
\begin{equation}
\tri{w^t-f^t,c^t}=\tri{w^t-\widetilde{w}^{t+1},c^t}+\tri{\widetilde{w}^{t+1}-f^t,c^t}.
\end{equation}
For the first term we may reuse the following bound from the proof of Proposition \ref{prop:hedge}:
\begin{equation}\label{eq:noisy_term1_new}
\tri{w^t-\widetilde{w}^{t+1},c^t}\leq \eta \tri{w^t,c^t}.
\end{equation}
For the second term, as in Proposition \ref{prop:hedge}, we use the inequality:
\begin{align*}
\tri{\widetilde{w}^{t+1}-f^t,c^t} &= \frac{1}{\eta}\prn*{D_{R}(f^t\mid{}w^t) - D_R(f^t\mid{}\wt^{t+1}) - D_R(\wt^{t+1}\mid{}w^t)}\\
&\leq \frac{1}{\eta} \prn*{D_R(f^t|w^t)-D_R(f^t|g^{t+1})},
\end{align*}
where the Bregman divergence is the KL divergence, i.e. $D_R(f|g)=\sum_j f_j \log(f_j/g_j)$.
Summing over all $t$, we have:
\begin{equation}
\sum_{t=1}^T \tri{w^t-f^t,c^t}\leq \eta \sum_{t=1}^T \tri{w^t,c^t}+\frac{1}{\eta} \sum_{t=1}^T \prn*{D_R(f^t|w^t)-D_R(f^t|g^{t+1})}
\end{equation}
To bound the second term, we distinguish between three cases. First, the term $D_R(f^1|w^1)$ can be bounded as in Proposition \ref{prop:hedge} by $\log(d)$ since $w^1$ is the uniform distribution. 

Second, at some $t>1$ where a change in the comparator occurred ($f^t\neq f^{t-1}$), we can bound $D_R(f^t|w^t)$ by $\log(d/\theta)$ since $w^t$ has is at least $\theta/d$ due to the mixing of the noise. This is exactly the reason why we need the noise --- this term could be unbounded otherwise.

Last, for some $t>1$ when the comparator did not change ($f^t=f^{t-1}$), we bound $D_R(f^t|w^t)-D_R(f^{t-1}|g^t)$ by $\theta\cdot d$. To prove that, note that since without loss of generality $f^{t}$ is an indicator vector, there is only one summand we are interested in the Bregman divergence. Let's call this summand $j$. What we want to bound is hence 
$$D_R(f^t|w^t)-D_R(f^{t-1}|g^t)=\log(1/w^t_j)-\log(1/g^t_j)=\log(g^t_j/w^t_j).
$$ 

As a result:
\begin{equation}\label{eq:noisy_term2_new}
\sum_{t=1}^T \Big(D_R(f^t|w^t)-D_R(f^{t}|g^{t+1}\Big)\leq \log(d)+ T\theta\log(d)+K \log(d/\theta).
\end{equation}
Combining inequalities \ref{eq:noisy_term1_new} and \ref{eq:noisy_term2_new} and setting $\theta=1/T$, the result follows.
\end{proof}

%% file: app_proof_dynamic_theorem.tex
% !TEX root = main.tex
The proof of Proposition \ref{thm:poa_dynamic} is analogous to that of Proposition \ref{prop:poa_full_info}.

Recall that $s^{*1:T}$ is a  solution sequence with cost at most $\rho$ times the minimum cost that is relatively stable to the turnover of players and that this sequence can be randomized.

For such a sequence of solutions, we use $K_i$ to denote the 
sum of total variation distances $K_i=\sum_t\nrm*{s_i^{*t}-s_i^{*t-1}}_1$ of the strategy for player $i$ in this sequence. 
\begin{eqnarray*}
(1-\epsilon) \sum_{t=1}^{T} \En\brk*{C(s^t)}&=&(1-\epsilon)\sum_{i\in\brk{d}} \sum_{t=1}^{T} \En\brk*{\cost_i(s^t)}\\
& \le & \sum_{i\in\brk{d}} \brk*{\sum_{t=1}^T \En\brk*{\cost_i(s_i^{*t},s_{-i}^t)}+ \frac{1+\En\brk*{K_i}}{\epsilon}A(d,T)}\\
&\le & \sum_{t=1}^{T} (\lambda\En\brk*{C(s^{*t})}+\mu\En\brk*{C(s^t)})+\frac{n+\En\brk*{\sum_i K_i}}{\epsilon}A(d,T).
\end{eqnarray*}
Here we are taking expectation over the randomness in $s_i^{*1:T}$ due to players turning and/or due to randomness in the approximate minimization algorithm.
The first inequality holds because each player satisfies the Low Approximate Regret property  (\ref{eq:lar_tv}) for total variation distance, applied with $s_i^{*1:T}$ as the comparator sequence. As was discussed in Appendix \ref{turnover_results_dynamic_games}, the property \eqref{eq:lar_tv} is implied by Low Approximate Regret for shifting experts \eqref{eq:lar_shifting}.
The second inequality follows from smoothness.

The claimed bound follows by rearranging terms.
\qed

%% file: app-utility.tex
In this section, we show how all our results extend to utility maximization games and mechanisms.

Consider a static game $G$ among a set of $n$ players.  Each player $i$ has an action space $S_i$ and a utility function $\utility_i: S_1\times\dots\times S_n\rightarrow [0,1]$ that maps an action profile $s=(s_1,\dots,s_n)$ to a utility $\utility_i(s)$. The goal of each player is to maximize their utility. One can simply adapt our definitions of Low Approximate Regret by treating utilities as negative costs. While one might imagine applying the same strategy to adapt algorithms to the utility setting, extra care is required. Not all algorithms necessarily admit such a direct adaptation (or adapt at all). However, all the algorithms analyzed in this paper do, and their proofs are designed to carry through with this adaptation. We demonstrate this by sketching the proofs for Hedge and Algorithm \ref{alg:bandit_alg} of Low Approximate Regret with utilities, but the same holds for all the other algorithms we analyze.
%, but as we will see in this section as Hedge and Algorithm \ref{alg:bandit_alg} do indeed satisfy Low Approximate Regret with utilities.

As in the cost minimization setting, we assume that at each round $t$, player $i$ picks a probability distribution $w_i^t$ and draws her action $s_i^t$ from this distribution. The utility she receives when playing action $x$ is $u_{i,x}^t=\utility_i(x,s_{-i}^t)$ where $s_{-i}^t$ is the set of strategies of all but $i^{\text{th}}$ player. Let $u_i^t=(u_{i,x}^t)_{x\in S_i}$.

An important class of utility maximization games are     mechanisms, such as auctions, where money plays special role. The players' actions $s_i$ typically involve bidding on items, and the outcome of an action profile $s$ comes in two parts: $v_i: S_1\times\dots\times S_n\rightarrow [0,1]$, which is the resulting value for player $i$, and $p_i:  S_1\times\dots\times S_n\rightarrow [0,1]$, which is the price player $i$ has to pay. Her utility is then $\utility_i(s)=v_i(s)-p_i(s)$.\footnote{We assume that all $s$ have $v_i(s)-p_i(s)\geq{}0$.}. 
We evaluate such mechanisms via the notion of social welfare $SW(s)=\sum_i v_i(s)$, the sum of the utilities of the players plus all the payments; this is the revenue of the mechanism. A simple example of such a mechanism is the first price auction: The player's strategy is a bid, and the highest bidder wins the item and pays her own bid.   

We use the the smooth mechanism definition of \cite{SyrgkanisT_STOC13}.\footnote{For the dynamic population game setting we use a variant of this definition, \textit{solution-based smoothness}, where $\opt$ in the RHS is replaced by the social welfare of a near-optimal solution as in \cite{LykourisST16}.} \begin{definition}[Smooth mechanism \cite{SyrgkanisT_STOC13}]
\textit{A utility maximization mechanism is called $(\lambda,\mu)$-smooth if there exists a strategy profile $s^{\star}$, such that
for all strategy profiles $s$: $\sum_i u_i(s_i^{\star},s_{-i})\geq \lambda \opt-\mu \sum_i p_i(s)$},
where $\opt=\max_{s^o} \sum_{i=1}^n \utility_i(s^o)$.
\end{definition}
Note the slight difference from Definition \ref{def:smoothness}. In proving the price of anarchy property we used the game's smoothness property with $s^*$ as the action profile resulting in $\opt$ total cost. In the definition for mechanisms, we do not insist that $SW(s^*)=\opt$.

Recall from section \ref{sec:model} that 
first price item auctions are $(1-1/e,1)$-smooth and all-pay actions are $(1/2,1)$-smooth. We  show in Proposition \ref{prop:poa-u_full_info} that smooth mechanisms have a price of anarchy of at most $\max(\mu,1)/\lambda$.

\begin{definition}[Low Approximate Regret for utility maximization]
A learning algorithm for player $i$ that uses action distributions $w_i^t$ in step $t$ satisfies the
Low Approximate Regret property for a parameter $\epsilon$, and a function $A(d,T)$ if for all action distributions $f\in\Delta(S_i)$:
\begin{equation}
\label{eq:lar-u}
(1+\epsilon)\sum_{t=1}^T \tri{w_i^t,u_i^t} \geq  \sum_{t=1}^T \tri{f,u_i^t}-\frac{A(d,T)}{\epsilon}.
\end{equation}

An algorithm satisfies Low Approximate Regret for the shifting experts setting if for all sequences  $f^1,\dots, f^T\in\Delta(S_i)$, letting $K$ be the number of shifts, i.e. $K=|\{t>2: f^{t-1}\neq f^t \}|$:
\begin{equation}
\label{eq:lar-u_shifting}
(1+\epsilon)\sum_{t=1}^T \tri{w_i^t,u_i^t} \geq  \sum_{t=1}^T \tri{f^t,u_i^t}-(1+K)\frac{A(d,T)}{\epsilon}.
\end{equation}
We say that an algorithm satisfies the \textit{Strong Low Approximate Regret property} if it satisfies \eqref{eq:lar-u} or \eqref{eq:lar-u_shifting} for all $\epsilon>0$ simultaneously. In the bandit feedback case, we require the property to hold in expectations over the realized strategies of player $i$.
\end{definition}

Now we are ready to prove the utility maximization analog of Proposition \ref{prop:poa_full_info}

\begin{proposition}[Efficiency for Mechanisms]
\label{prop:poa-u_full_info}
Consider a $(\lambda,\mu)$-smooth mechanism.  If all players use Low Approximate Regret algorithms satisfying Eq. (\ref{eq:lar-u}) for parameter $\epsilon$, then
$$
\frac{1}{T}\sum_t\En\brk*{SW(s^t)}\geq \frac{\lambda}{\max(\mu,1+\epsilon)}\opt+\frac{n}{T}\cdot\frac{1}{\max(\mu,1+\epsilon)}\cdot \frac{A(d,T)}{\epsilon}.
$$
where $s^t$ is the action profile drawn on round $t$ from the corresponding mixed actions of the players.
\end{proposition}
\begin{proof}
 We get the claimed bound  by considering
$(1+\epsilon)\sum_t\En\brk*{\sum_i \utility_i(s^t)}$, using the low approximate regret property with $f=s^{\star}_i$ for each player $i$ for the action $s^{\star}$in the smoothness property, then using the smoothness property for each time $t$ to bound  $\sum_i \utility_i(s_i^{\star},s_{-i}^t)$, and rearranging terms.
\end{proof}

\begin{proposition}
\textit{Hedge with a constant learning rate and uniform prior over actions satisfies the utility version of the Low Approximate Regret property with $A(d,T)=(e-1)\log(d)$.}
\end{proposition}
We mirror the proof of Proposition \ref{prop:hedge} with $c^t=-u^t$. The only place where the analysis does not automatically go through is where we need that the costs are in $[0,1]$, namely equation \eqref{eq:hedge_loss_place}. Note that the first inequality there ceases to hold when $c^t<0$. However it is still the case that $1-e^{-\eta x}\leq (e-1)\eta x$ for $x\in[-1,0]$. Hence we have:
\begin{equation*}
    \tri{w^t-\widetilde{w}^{t+1},c^t}\leq \eta (e-1)\sum_{j\in[d]} w_j^t (c_j^t)^2\leq \eta (e-1)\sum_{j\in[d]} w_j^t (-c_j^t).
\end{equation*}
The last inequality holds as $-c_j^t\in[0,1]$.

With this inequality, combined with the rest of the proof in Proposition \ref{prop:hedge}, we have:
$\sum_t\tri{w_t-f,c^t}\leq \eta (e-1)\sum_t\tri{w^t,-c^t}+\frac{\log(d)}{\eta}$. Setting $\epsilon=\eta(e-1)$ and substituting $c^t$ yields
\begin{equation*}
(1+\epsilon)\sum_{t=1}^{T}\tri{w^t,u^t}\geq \sum_{t=1}^{T}\tri{f,u^t}-\frac{(e-1)\log(d)}{\epsilon}.
\end{equation*}
which proves the claim.

\paragraph{Bandit Feedback}
We now provide some more discussion regarding Algorithm \ref{alg:bandit_alg}, since the improvement on the number of strategies occurs in utlility maximization settings.

The algorithm's update step for utilities is obtained by using $c^t=-u^t$, but note that the normalization factor is $\gamma\geq{}0$ for utility settings.

The Low Approximate Regret proof is achieved as in Lemma \ref{lem:bandit_static} again by replacing cost with negative utility.

\begin{lemma}[Regret bound for Algorithm \ref{alg:bandit_alg} with utilities]
\label{lem:bandit_final_utilities}
For any $f^{\star}\in\Delta(d)$, and any sequence of utilities $u^{1},\ldots,u^{T}\in\brk*{0,1}^{d}$, the weights generated by Algorithm \ref{alg:bandit_alg} with $\eta\in\prn*{0,1}$ satisfy
\
\begin{equation}\label{eq:bandit_final_utilities}
\En\brk*{\sum_{t=1}^T \tri{w^t,u^t}} \geq \En\brk*{\sum_{t=1}^T \tri{f^{\star},u^t}} -\frac{\eta}{1 -\eta} \En\brk*{\sum_{t=1}^T \tri{w^t,u^t}}
    -\frac{1}{\eta}d\log(T/d) - d.
\end{equation}
\end{lemma}
\begin{proof}[Proof of Lemma \ref{lem:bandit_final_utilities}]
Define a cost sequence $c^1,\ldots,c^T$ via $c^t=-u^t$ and run Algorithm \ref{alg:bandit_alg} with these costs. From Lemma \ref{lem:bandit_full_info}, we have that for each $t$,
\[
-\tri{w^t-f^{\star},\hat{u}^t}  \leq  \eta\sum_{j\in\brk{d}} \frac{(w_j^t \cdot \hat{u}_j^t)^2}{1+\eta w_j^t \hat{u}_j^t}
    +\frac{1}{\eta}\prn*{D_R\prn*{f^{\star}\mid{}w^t}-D_R\prn*{f\mid{}w^{t+1}}},
\]
where $\hat{u}^t_j=\ind\crl*{j=s^t}u^t_j/w^t_j$.
Applying an analysis identical to that of Lemma \ref{lem:bandit_final} on this bound yields the result.
\end{proof}

%% file: main.bbl
\newcommand{\etalchar}[1]{$^{#1}$}
\begin{thebibliography}{AAGO06}

\bibitem[AAGO06]{Allenberg2006}
Chamy Allenberg, Peter Auer, L{\'a}szl{\'o} Gy{\"o}rfi, and Gy{\"o}rgy
  Ottucs{\'a}k.
\newblock {\em Hannan Consistency in On-Line Learning in Case of Unbounded
  Losses Under Partial Monitoring}, pages 229--243.
\newblock Springer Berlin Heidelberg, Berlin, Heidelberg, 2006.

\bibitem[AB10]{Alletal06}
Jean-Yves Audibert and S{\'e}bastien Bubeck.
\newblock Regret bounds and minimax policies under partial monitoring.
\newblock {\em The Journal of Machine Learning Research}, 11:2785--2836, 2010.

\bibitem[ACBG02]{Auer2002Adaptive}
Peter Auer, Nicolo Cesa-Bianchi, and Claudio Gentile.
\newblock Adaptive and self-confident on-line learning algorithms.
\newblock {\em Journal of Computer and System Sciences}, 64(1):48--75, 2002.

\bibitem[AHR08]{Abernethy08competingin}
Jacob Abernethy, Elad Hazan, and Alexander Rakhlin.
\newblock Competing in the dark: An efficient algorithm for bandit linear
  optimization.
\newblock In {\em Proc. of the 21st Annual Conference on Learning Theory
  (COLT)}, 2008.

\bibitem[BDH{\etalchar{+}}08]{Bartlett2008High}
Peter~L Bartlett, Varsha Dani, Thomas Hayes, Sham Kakade, Alexander Rakhlin,
  and Ambuj Tewari.
\newblock High-probability regret bounds for bandit online linear optimization.
\newblock In {\em Proceedings of 21st Annual Conference on Learning Theory
  (COLT)}, pages 335--342, 2008.

\bibitem[CBL06]{prediction_book}
Nicolo Cesa-Bianchi and Gabor Lugosi.
\newblock {\em Prediction, Learning, and Games}.
\newblock Cambridge University Press, New York, NY, USA, 2006.

\bibitem[CBMS07]{Cesa2007Improved}
Nicolo Cesa-Bianchi, Yishay Mansour, and Gilles Stoltz.
\newblock Improved second-order bounds for prediction with expert advice.
\newblock {\em Machine Learning}, 66(2-3):321--352, 2007.

\bibitem[CK05]{Christodulou}
Giorgos Christodoulou and Elias Koutsoupias.
\newblock The price of anarchy of finite congestion games.
\newblock In {\em Proceedings of the 37th Annual ACM Symposium on Theory of
  Computing (STOC)}, pages 67 -- 73, 2005.

\bibitem[DDK15]{Daskalakis2015Near}
Constantinos Daskalakis, Alan Deckelbaum, and Anthony Kim.
\newblock Near-optimal no-regret algorithms for zero-sum games.
\newblock {\em Games and Economic Behavior}, 92:327--348, 2015.

\bibitem[FRS15]{Foster2015Adaptive}
Dylan~J Foster, Alexander Rakhlin, and Karthik Sridharan.
\newblock Adaptive online learning.
\newblock In {\em Advances in Neural Information Processing Systems}, pages
  3357--3365, 2015.

\bibitem[FS97]{Freund97}
Yoav Freund and Robert~E Schapire.
\newblock A decision-theoretic generalization of on-line learning and an
  application to boosting.
\newblock {\em J. Comput. Syst. Sci.}, 55(1):119--139, August 1997.

\bibitem[Haz16]{Hazan_book}
Elad Hazan.
\newblock {\em Introduction to Online Convex Optimization}.
\newblock Foundations and Trends in Optimization, 2016.

\bibitem[HK10]{Hazan2010Extracting}
Elad Hazan and Satyen Kale.
\newblock Extracting certainty from uncertainty: Regret bounded by variation in
  costs.
\newblock {\em Machine learning}, 80(2-3):165--188, 2010.

\bibitem[HS09]{Hazan2009}
Elad Hazan and C.~Seshadhri.
\newblock Efficient learning algorithms for changing environments.
\newblock In {\em Proceedings of the 26th Annual International Conference on
  Machine Learning (ICML)}, pages 393--400, 2009.

\bibitem[KKL09]{KakadeSTOC07}
Sham~M. Kakade, Adam~Tauman Kalai, and Katrina Ligett.
\newblock Playing games with approximation algorithms.
\newblock {\em SIAM J. Comput.}, 39:1088 -- 1106, 2009.

\bibitem[KP09]{Koutsoupias2009Worst}
Elias Koutsoupias and Christos Papadimitriou.
\newblock Worst-case equilibria.
\newblock {\em Comp. sci. review}, 3(2):65--69, 2009.

\bibitem[KVE15]{Koolen2015Second}
Wouter~M Koolen and Tim Van~Erven.
\newblock Second-order quantile methods for experts and combinatorial games.
\newblock In {\em Proceedings of The 28th Conference on Learning Theory
  (COLT)}, pages 1155--1175, 2015.

\bibitem[LS15]{Luo2015}
Haipeng Luo and Robert~E Schapire.
\newblock Achieving all with no parameters: Adanormalhedge.
\newblock In {\em Proceedings of The 28th Conference on Learning Theory
  (COLT)}, pages 1286--1304, 2015.

\bibitem[LST16]{LykourisST16}
Thodoris Lykouris, Vasilis Syrgkanis, and {\'E}va Tardos.
\newblock Learning and efficiency in games with dynamic population.
\newblock In {\em Proceedings of the Twenty-Seventh Annual ACM-SIAM Symposium
  on Discrete Algorithms (SODA)}, pages 120--129. SIAM, 2016.

\bibitem[Neu15]{Neu15}
Gergely Neu.
\newblock First-order regret bounds for combinatorial semi-bandits.
\newblock In {\em {Proceedings of the 27th Annual Conference on Learning Theory
  (COLT)}}, pages 1360--1375, 2015.

\bibitem[Rou15]{Roughgarden15}
Tim Roughgarden.
\newblock Intrinsic robustness of the price of anarchy.
\newblock {\em Journal of the ACM}, 2015.

\bibitem[RS13a]{RakhlinS13predictablesequences}
Alexander Rakhlin and Karthik Sridharan.
\newblock Online learning with predictable sequences.
\newblock In {\em Conference on Learning Theory (COLT)}, pages 993--1019, 2013.

\bibitem[RS13b]{Rakhlin2013Optimization}
Alexander Rakhlin and Karthik Sridharan.
\newblock Optimization, learning, and games with predictable sequences.
\newblock In {\em Advances in Neural Information Processing Systems (NIPS)},
  pages 3066--3074, 2013.

\bibitem[RST16]{AuctionSurvey2016}
Tim Roughgarden, Vasilis Syrgkanis, and Eva Tardos.
\newblock The price of anarchy in auctions.
\newblock Available at {\small \url{https://arxiv.org/abs/1607.07684}}, 2016.

\bibitem[RT02]{RoughgardenT2002}
Tim Roughgarden and Eva Tardos.
\newblock How bad is selfish routing?
\newblock {\em Journal of the ACM}, 49:236 -- 259, 2002.

\bibitem[SALS15]{SyrgkanisALS15}
Vasilis Syrgkanis, Alekh Agarwal, Haipeng Luo, and Robert~E Schapire.
\newblock Fast convergence of regularized learning in games.
\newblock In {\em Advances in Neural Information Processing Systems (NIPS)},
  pages 2989--2997, 2015.

\bibitem[SL14]{Steinhardt2014Adaptivity}
Jacob Steinhardt and Percy Liang.
\newblock Adaptivity and optimism: An improved exponentiated gradient
  algorithm.
\newblock In {\em Proceedings of the 31st International Conference on Machine
  Learning (ICML)}, pages 1593--1601, 2014.

\bibitem[ST13]{SyrgkanisT_STOC13}
Vasilis Syrgkanis and \'{E}va Tardos.
\newblock Composable and efficient mechanisms.
\newblock In {\em ACM Symposium on Theory of Computing (STOC)}, pages 211--220,
  2013.

\bibitem[Sto05]{Stoltz05}
Gilles Stoltz.
\newblock Incomplete information and internal regret in prediction of
  individual sequences.
\newblock {\em PhD thesis, Universite Paris-Sud}, 2005.

\bibitem[YEYS04]{Yaroshinsky2004}
Rani Yaroshinsky, Ran El-Yaniv, and Steven~S. Seiden.
\newblock How to better use expert advice.
\newblock {\em Machine Learning}, 55(3):271--309, 2004.

\end{thebibliography}
